\documentclass[11pt,letterpaper]{article}

\usepackage[utf8]{inputenc}
\usepackage[english]{babel}

\usepackage[dvipsnames,usenames]{xcolor}
\usepackage[colorlinks=true,pdfpagemode=UseNone,urlcolor=RoyalBlue,linkcolor=RoyalBlue,citecolor=OliveGreen,pdfstartview=FitH]{hyperref}

\usepackage{tikz}
\usetikzlibrary{shapes.geometric}

\usepackage{subcaption}

\usepackage[margin=1in]{geometry}
\usepackage{amsfonts,amsmath,amsthm,thmtools}
\usepackage{nicefrac}

\declaretheorem[name=Informal Theorem]{informal}
\declaretheorem{theorem}
\declaretheorem[sibling=theorem]{lemma}
\declaretheorem[sibling=theorem]{corollary}

\usepackage{todonotes}

\usepackage{booktabs}
\usepackage{caption}
\usepackage[numbers, sort]{natbib}

\usepackage{tcolorbox}

\usepackage{enumitem}

\newcommand{\E}{\mathbf{E}}
\renewcommand{\Pr}{\mathbf{Pr}}

\newcommand{\alg}{A}
\newcommand{\opt}{\textsc{Opt}}
\newcommand{\jl}{\textsc{JL}}
\newcommand{\algo}{\textsc{Alg}}

\newcommand{\defeq}{\overset{\text{def}}{=}}

\newcommand{\Sref}[1]{\hyperref[#1]{\S\ref*{#1}}}

\newcommand{\R}{\mathbb{R}}
\newcommand{\Z}{\mathbb{Z}}

\newcommand{\Dx}{{\Delta_x}}
\newcommand{\Dt}{{\Delta_t}}

\title{The Power of Multiple Choices in Online Stochastic Matching}

\author{
    Zhiyi Huang
    \thanks{The University of Hong Kong. Email: zhiyi@cs.hku.hk.}
    \and
    Xinkai Shu
    \thanks{The University of Hong Kong. Email: xkshu@cs.hku.hk.}
    \and
    Shuyi Yan
    \thanks{Tsinghua University. Email: yansy18@mails.tsinghua.edu.cn.}
}
\date{March 2022}

\begin{document}

\begin{titlepage}
    \thispagestyle{empty}
    \maketitle
    \begin{abstract}
        \thispagestyle{empty}
        We study the power of multiple choices in online stochastic matching.
Despite a long line of research, existing algorithms still only consider two choices of offline neighbors for each online vertex because of the technical challenge in analyzing multiple choices.
This paper introduces two approaches for designing and analyzing algorithms that use multiple choices.
For unweighted and vertex-weighted matching, we adopt the online correlated selection (OCS) technique into the stochastic setting, and improve the competitive ratios to $0.716$, from $0.711$ and $0.7$ respectively.
For edge-weighted matching with free disposal, we propose the Top Half Sampling algorithm.
We directly characterize the progress of the whole matching instead of individual vertices, through a differential inequality.
This improves the competitive ratio to $0.706$, breaking the $1-\frac{1}{e}$ barrier in this setting for the first time in the literature. 
Finally, for the harder edge-weighted problem without free disposal, we prove that no algorithms can be $0.703$ competitive, separating this setting from the aforementioned three.

    \end{abstract}
\end{titlepage}

\section{Introduction}
\label{sec:intro}

Optimization problems in real-life scenarios often need to consider the uncertainties of what will happen in future.
In the multibillion business of online advertising, for example, a platform receives over time user queries, a.k.a., impressions, and selects advertisements to show immediately when each impression arrives without accurate knowledge of what impressions will come later.

If the algorithm knows nothing about future impressions, this is modeled by the online bipartite matching problem of \citet*{KarpVV:STOC:1990} and its variants.
Advertisers and impressions are vertices on the two sides of a bipartite graph.
Advertisers are known to the platform beforehand, and therefore are called the \emph{offline vertices}.
Impressions arrive over time and are the \emph{online vertices}.
Edges of the bipartite graph model which advertiser is interested in which impression, e.g., based on the search keywords.
When each online vertex arrives, the algorithm sees its incident edges and irrevocably decides how to match it right away.
\citet{KarpVV:STOC:1990} studied unweighted matching and gave a $1-\frac{1}{e}$ competitive Ranking algorithm.
It means that the expected size of the online matching is at least $1-\frac{1}{e}$ times the maximum matching size in hindsight.
They also proved it to be the best possible under worst-case analysis.

Online advertising platforms, however, have a lot of data which provide information about future impressions.
Can we use the information to break the $1-\frac{1}{e}$ barrier?
Modeling the information as a distribution from which each online impression is drawn, \citet{FeldmanMMM:FOCS:2009} proposed the online stochastic matching model and gave an affirmative answer to this question assuming integral arrival rates of online vertices.%
\footnote{That is, the expected number of each type of online vertex realized in the bipartite graph is a positive integer.}
Their algorithm builds on an idea called \emph{the power of two choices}, originally from online load balancing~\cite{AzarBKU:STOC:1994, Mitzenmacher:TPDS:2001}.
For each online vertex, the algorithm finds two candidate offline neighbors based on the distribution of online vertices, \emph{independent of which offline vertices have been matched}; it then matches the online vertex to the first unmatched candidate if there is any.
A long line of subsequent research not only improved the competitive ratio, but also generalized the results to vertex-weighted matching~\cite{BrubachSSX:EDA:2016, JailletL:MOR:2014, HuangS:STOC:2021} and edge-weighted matching~\cite{HaeuplerMZ:WINE:2011, BrubachSSX:EDA:2016}, and to arbitrary arrival rates~\cite{ManshadiOS:MOR:2012, JailletL:MOR:2014, HuangS:STOC:2021}.
Despite much progress, these improved algorithms are still under the above framework of the power of two choices.

%To further advance the state-of-the-art, 
This paper asks two natural questions:
(1) \emph{Why not examine multiple choices for each online vertex?}
(2) \emph{Why not focus on unmatched neighbors when selecting candidates?}
To our knowledge, previous algorithms did not make these obvious improvements mainly because of the difficulties in the analysis.
Existing analyses rely on fine-grained characterizations of the match probabilities of offline vertices.
On the one hand, restricting to two choices implies that each offline vertex can only be matched by two means: as the first candidate, or as the second candidate when the first is already matched.
The analyses then consider the probabilities of the two cases separately.
On the other hand, selecting candidates independent of past matched events ensures that the match events of different offline vertices are almost independent (e.g., \cite{JailletL:MOR:2014, HuangS:STOC:2021}).
The independence substantially simplifies the calculations of the offline vertices' matched probabilities.

\subsection{Our Contributions}
This paper proposes two approaches to overcome the above difficulties.
The first approach adopts a recent technique called \emph{online correlated selection} (OCS)~\cite{FahrbachHTZ:FOCS:2020, GaoHHNYZ:FOCS:2021, ShinA:arXiv:2021, BlancC:FOCS:2021} into the stochastic setting, so that we can make adaptive matching decisions with only unmatched neighbors as candidates, and can still analyze the matched probabilities of offline vertices.
The second approach sidesteps the characterization of individual offline vertices' matched probabilities, and considers the online algorithm's overall progress directly.
Both approaches lead to improved algorithms and competitive ratios in various settings of online stochastic matching.
We next briefly introduce our algorithms and techniques in the Poisson arrival model, in which online vertices of each type independently follow a Poisson process with time horizon $[0, 1]$.
This is asymptotically equivalent to the original online stochastic matching model~\cite{HuangS:STOC:2021}.
Table~\ref{tab:summary} gives a summary of our improved competitive ratios for general arrival rates;
see Subsection~\ref{sec:related-work} for previous results that assumed integral arrival rates.

\begin{table}[t]
    \renewcommand{\thefootnote}{\fnsymbol{footnote}}
    \renewcommand{\arraystretch}{1.1}
    \centering
    \caption{Summary of Results. We round down algorithmic results and round up hardness results, to three decimal places. The results in this paper are on the right of the arrows in bold.}
    \label{tab:summary}
    \begin{tabular}{l@{\hskip 24pt}l@{\hskip 6pt}l@{\hskip 6pt}l@{\hskip 24pt}l@{\hskip 6pt}l@{\hskip 6pt}l}
        \toprule
        & \multicolumn{3}{c}{Algorithms} & \multicolumn{3}{c}{Hardness} \\
        \midrule
        Unweighted & $0.711$~\cite{HuangS:STOC:2021} & $\to$ & $\mathbf{0.716}$~\textbf{(\Sref{sec:vertex-weighted-ocs})} & $0.823$~\cite{ManshadiOS:MOR:2012} \\
        Vertex-weighted & $0.700$~\cite{HuangS:STOC:2021} & $\to$ & $\mathbf{0.716}$~\textbf{(\Sref{sec:vertex-weighted-ocs})} & $0.823$~\cite{ManshadiOS:MOR:2012} \\
    Edge-weighted (Free Disposal) & $0.632$~\cite{FeldmanMMM:FOCS:2009} \footnotemark[2] & $\to$ & $\mathbf{0.706}$~\textbf{(\Sref{sec:edge-weighted-differential-inequality})} & $0.823$~\cite{ManshadiOS:MOR:2012} \\
        Edge-weighted & $0.632$~\cite{FeldmanMMM:FOCS:2009} \footnotemark[2] & & & $0.823$~\cite{ManshadiOS:MOR:2012} & $\to$ & $\mathbf{0.703}$~\textbf{(\Sref{sec:hardness-edge-weighted})}\\
        \bottomrule
        \multicolumn{7}{p{16cm}}{\footnotesize\footnotemark[2] Although \citet{FeldmanMMM:FOCS:2009} only analyzed unweighted matching, they effectively showed that every edge is matched with probability $1-\nicefrac{1}{e}$ times the LP variable, which is sufficient for edge-weighted matching as well.}
    \end{tabular}
\end{table}

\paragraph{Poisson Online Correlated Selection.}
Like previous works, we rely on the optimal solution of a linear program (LP) relaxation.
For any online vertex type $i$ and offline vertex $j$, the LP gives $0 \le x_{ij} \le 1$ as a reference of how likely a type $i$ online vertex should be match to $j$.
Intuitively, we may want to match a type $i$ online vertex to each unmatched neighbor $j$ with probability proportional to $x_{ij}$.
Inspired by the multi-way semi-OCS of \citet{GaoHHNYZ:FOCS:2021}, we propose the Poisson OCS which further adjusts the match probabilities based on the LP matched level $x_j = \sum_i x_{ij}$ of offline vertices $j$.
If a type $i$ online vertex arrives at time $t$, Poisson OCS matches it to each unmatched neighbor $j$ with probability proportional to $e^{t x_j} x_{ij}$.
The exponential weights come from an informal invariant.
If we match each type $i$ online vertex independently to a neighbor $j$ (matched or not) with probability $x_{ij}$, the probability that an offline vertex $j$ is still unmatched at time $t$ equals $e^{-t x_j}$.
Since Poisson OCS is better, this unmatched probability is at most $e^{-t x_j}$.
The expected mass of matching $i$ to $j$ at time $t$ in Poisson OCS is therefore at most $e^{t x_j} x_{ij} \cdot e^{-t x_j} = x_{ij}$.

\paragraph{Poisson Matching Linear Program Hierarchy.}
Our Poisson OCS analysis requires that the LP does not match the online vertices close to deterministically.
If every online vertex type $i$ is fully matched to a single offline neighbor, the competitive ratio is only $1-\frac{1}{e}$.
Fortunately, existing LPs in the literature satisfy this requirement.
We first study the Natural LP and a corresponding Converse Jensen Inequality of \citet{HuangS:STOC:2021}, which already allow us to improve the competitive ratio of vertex-weighted online stochastic matching from $0.7$ to $0.707$.
Moreover, we introduce a Poisson Matching LP Hierarchy that contains the Natural LP at its first level.
We give a polynomial-time separation oracle for the LP at any constant level, and thus show its computational tractability.
Finally, we prove a Converse Jensen Inequality for the second level Poisson Matching LP, and get the following result for unweighted and vertex-weighted matching.

\begin{informal}
    Poisson OCS with the second level Poisson Matching LP is a polynomial-time, $0.716$-competitive algorithm for unweighted and vertex-weighted online stochastic matching.
\end{informal}

\paragraph{Top Half Sampling.}
We next turn to edge-weighted matching.
First consider the \emph{free disposal} model which allows the algorithm to rematch a matched offline vertex to a heavier edge.
In online advertising, it means that we show an advertisement multiple times, but only charge the advertiser for the most valuable impression.
We design the Top Half Sampling algorithm.
When a type $i$ online vertex comes, first consider the marginal weight of matching it to each offline neighbor $j$:
if $j$ is currently matched to an edge with weight $w$, the marginal equals $\max \{w_{ij} - w, 0 \}$.
Next, we sort the offline neighbors by descending order of the marginals, double the LP matched probability $x_{ij}$, and truncate the excessive probabilities if the total exceeds $1$.
For example, if there are three neighbors sorted by the marginals and the $x_{ij}$'s equal $0.4$, $0.4$, and $0.2$, then the adjusted probabilities are $0.8$, $0.2$, and $0.0$.
That is, the algorithm focuses on the most valuable half of the neighbors.
Finally, match to an offline neighbor according to the adjusted probabilities.

\paragraph{Whole-match Analysis via Differential Inequality.}
Top Half Sampling biases towards some of the offline vertices by definition.
Therefore, some offline vertices are matched faster than what the LP solution indicates, and the others are matched slower.
This calls for an analysis that considers the progress of the whole matching instead of individual offline vertices.
Indeed, our analysis will characterize the whole-match progress by a linear differential inequality.
Let $\bar{\alg}(t)$ be the difference between the LP optimal and the algorithm's expected objective at any time $t$.
We will establish $c_0 \bar{\alg}(t) + c_1 \frac{d}{dt} \bar{\alg}(t) + c_2 \frac{d^2}{dt^2} \bar{\alg}(t) \le 0$ for some positive coefficients $c_0, c_1, c_2$.
Intuitively, it means that the farther the algorithm's objective is from the LP optimal (i.e., if $\bar{\alg}(t)$ is large), and the faster the marginal match rate decreases (i.e., if $\frac{d^2}{dt^2} \bar{\alg}(t)$ is large), the larger the current match rate would be (i.e., $\frac{d}{dt} \bar{\alg}(t)$ must be sufficiently negative).
This breaks the $1-\frac{1}{e}$ barrier in edge-weighted online stochastic matching with free disposal for the first time in the literature for arbitrary arrival rates. 

\begin{informal}
    \label{inf:edge-weighted}
    Top Half Sampling is a polynomial-time, $0.706$-competitive algorithm for edge-weighted online stochastic matching with free disposal.
\end{informal}

\paragraph{On Multiple Choices.}
Informally, for each online vertex, previous algorithms based on the power of two choices first sample two offline vertices according to the LP solution, \emph{oblivious to the previous matching decisions}.
These algorithms then match the online vertex to the first unmatched sampled offline vertex.
The natural extension is to sample more than two offline vertices in the first step.
In the limit when we sample an infinite number of offline vertices, it becomes sampling without replacement.
Poisson OCS is a sampling without replacement with adjusted marginals.
Top Half Sampling is also a sampling without replacement with adjusted marginals for any online type, unless the sum of LP variables corresponding to matched neighbors exceeds half the type's arrival rate.
In this sense, our algorithms use the power of multiple choices.

\paragraph{Hardness Results.}
For the harder edge-weighted model \emph{without free disposal}, we prove that no algorithms can achieve the competitive ratios in the above positive results.
This holds even if we have unlimited computational power.
Although this hardness does not rule out breaking $1-\frac{1}{e}$, it separates the edge-weighted problem without free disposal from the other three settings.

\begin{informal}
    No algorithm can be $0.703$-competitive for edge-weighted online stochastic matching without free disposal.
\end{informal}

Finally, we prove that no algorithm can be more than $0.706$%
\footnote{It is $1 - \frac{1}{1-\ln 2} \big( \frac{1}{2e} - \frac{\ln 2}{e^2} \big)$. We incorrectly round it down to stress it is the same ratio as in Informal Theorem~\ref{inf:edge-weighted}.}
competitive compared to the LP of \citet{JailletL:MOR:2014}.
Thus, more expressive LPs such as the Natural LP and Poisson Matching LPs are necessary for getting our improvements in unweighted and vertex-weighted matching.

\paragraph{Organization.}
We will first present Top Half Sampling in Section~\ref{sec:edge-weighted-differential-inequality} because of its simplicity, followed by Poisson OCS in Section~\ref{sec:vertex-weighted-ocs}.
These two sections are self-contained and therefore can be read separately.
Section~\ref{sec:hardness} will prove the hardness results.

\subsection{Other Related Work}
\label{sec:related-work}

The literature has obtained better competitive ratios when the online types have integral arrival rates.
The best known ratio for unweighted and vertex-weighted matching is $0.7299$, and the best result for the edge-weighted case is $0.705$, due to \citet*{BrubachSSX:EDA:2016}.
Therefore, our upper bound of $0.703$ for edge-weighted matching with general arrival rates shows that the general problem is strictly harder than the special case of integral arrival rates.
Further, \citet{BrubachSSX:EDA:2016} and \citet{HaeuplerMZ:WINE:2011} partially employed three choices in their algorithms for the special case of integral arrival rates.

The optimal Ranking algorithm of \citet{KarpVV:STOC:1990} 
for the worst-case model of online matching was generalized to vertex-weighted matching by \citet{AggarwalGKM:SODA:2011}.
\citet{FeldmanKMMP:WINE:2009} noted that no non-trivial competitive ratio is possible for edge-weighted online matching \emph{without free disposal}.
They introduced the free disposal model and gave an optimal $1-\frac{1}{e}$ competitive algorithm under a large-market assumption.
\citet{MehtaSVV:JACM:2007} introduced another notable variant of online matching called AdWords.
They gave an optimal $1-\frac{1}{e}$ competitive algorithm also assuming a large market.
These results were later simplified and unified under the online primal dual framework~\cite{DevanurJK:SODA:2013, BuchbinderJN:ESA:2007, DevanurHKMY:TEAC:2016}.
The aforementioned OCS technique led to the first algorithms that improve the $\frac{1}{2}$ competitive ratio of greedy in the general cases of edge-weighted matching~\cite{FahrbachHTZ:FOCS:2020, BlancC:FOCS:2021, ShinA:arXiv:2021, GaoHHNYZ:FOCS:2021} and AdWords~\cite{HuangZZ:FOCS:2020}.

The literature has also studied models between the worst-case and stochastic models, notably the unknown IID model and the random arrival model.
The former assumes that online vertices are IID like online stochastic matching, but the algorithm does not know the distribution.
The latter considers worst-case bipartite graphs like the worst-case model, but assumes a random arrival order of online vertices.
These models allow us to break the $1-\frac{1}{e}$ barrier in unweighted matching~\cite{MahdianY:STOC:2011, KarandeMT:STOC:2011} and vertex-weighted matching~\cite{HuangTWZ:TALG:2019, JinW:arXiv:2020}, to get a $\frac{1}{e}$ competitive edge-weighted algorithm without free disposal~\cite{KesselheimRTV:ESA:2013}, and even to be $1-\epsilon$ competitive in AdWords under a large-market assumption~\cite{DevanurJSW:JACM:2019}.

\section{Preliminaries}
\label{sec:prelim}

%\zh{This is a tentative preliminaries copied from our earlier notes.}

\paragraph{Poisson Arrival Model.}
Consider a set $I$ of online vertex types and a set $J$ of offline vertices.
Each online type $i$ arrives independently following a Poisson process with time horizon $0 \le t \le 1$ and arrival rate $\lambda_i$.
Define $\lambda_S = \sum_{i \in S} \lambda_i$ for any subset $S$ of online types.
We refer to $\Lambda = \lambda_I$ as the (expected) number of online vertices.
Online vertices of type $i$ are adjacent to a subset $J_i$ of offline vertices.
We denote the set of edges as $E = \{ (i, j) : i \in I, j \in J_i \}$, and let $I_j = \{ i \in I : (i, j) \in E \}$ be the set of online types adjacent to any offline vertex $j$.
Each edge $(i,j)$ has a positive weight $w_{ij}$.
We want to maximize the expected total weight of the matched edges.
The problem is \emph{vertex-weighted} if $w_{ij} = w_j$ for all edges $(i, j)$ in which case the objective becomes maximizing the expected total weight of the matched offline vertices.
It is \emph{unweighted} if $w_{ij} = 1$ for all edges $(i, j)$ in which case the objective is to maximize the expected cardinality of the matching.

\paragraph{(Original) Online Stochastic Matching Model.}
The model by \citet{FeldmanMMM:FOCS:2009} considers a fixed integral number of online vertices $\Lambda$.
The online vertices arrive one by one, each independently draws a type $i$ with probability $\frac{\lambda_i}{\Lambda}$.
The rest is the same as the Poisson arrival model.
\citet{HuangS:STOC:2021} proved an asymptotic equivalence between the two models.
%Hence we consider them interchangeably.
%Therefore, this paper will consider the two models interchangeably.

\paragraph{Online Algorithms.}
An online algorithm immediately and irrevocably decides how to match each online vertex when it arrives.
%An online algorithm makes the matching decision for each online vertex
%irrevocably and immediately upon its arrival.
%Let Alg be the expected total weight of the edges in the algorithm’s matching.
This paper considers the worst-case competitive analysis which examines the ratio of an online algorithm's expected objective to the expected optimal objective in hindsight, i.e., that from always choosing the optimal matching based on the realized bipartite graph.
%wwith respect to the expected total weight of the maximum weight matching of the realized graph G, denoted as Opt.
The competitive ratio of an online algorithm is the infimum of this ratio.
We also say that an online algorithm is $\Gamma$-competitive if its competitive ratio is at least $\Gamma$.

The problem admits \emph{free disposal} if each offline vertex can be matched multiple times, yet only the heaviest edge matched to it contributes to the objective.
This effectively allows the algorithm to dispose earlier lighter edges for free.
Note that free disposal is irrelevant in the vertex-weighted and unweighted problems since all edges adjacent to an offline vertex have the same weight.

\paragraph{Natural Linear Program.}
Like previous works, our analyses compare the algorithms' objectives to the optimal value of a linear program (LP) relaxation of the online stochastic matching problem.
Our starting point is the Natural LP introduced by \citet{HuangS:STOC:2021}
(see also \citet{TorricoAT:MP:2018} for a version for online stochastic matching with integral arrival rates), which is the tightest in the literature thus far:
\[
    \mbox{max}_{x \ge 0} \quad \sum_{(i,j)\in E} w_{ij} x_{ij} \quad \mbox{s.t.} \quad \forall i \in I: \sum_{j \in J_i} x_{ij}\leq\lambda_i ~,\quad \forall j \in J, \forall S \subseteq I_j : \sum_{i\in S} x_{ij} \leq 1-e^{-\lambda_S} 
    ~.
\]

%\begin{alignat}{2}
%    \mbox{maximize}\quad & \sum_{(i,j)\in E} w_{ij} x_{ij} & {} & \nonumber \\
%    \mbox{subject to}\quad & \sum_{j \in J_i} x_{ij}\leq\lambda_i &\quad& \forall i\in I \nonumber \\
    %\label{eqn:natural-lp}
%    & \sum_{i\in S} x_{ij} \leq 1-e^{-\lambda_S} &\quad& \forall j\in J, \forall S\subseteq I_j \\
%    & x_{ij} \geq 0 &\quad& \forall(i,j)\in E \nonumber
%\end{alignat}

Let $\opt$ be the optimal LP objective.
We abuse notation and let $x_{ij}$ also denote the optimal solution of the natural LP (and of the tighter LPs that we will introduce in this paper).
Further artificially define $x_{ij} = 0$ for any $(i, j) \notin E$ for notational simplicity.
Finally, let $\rho_{ij} = \frac{x_{ij}}{\lambda_i}$ which may be interpreted as the probability that an online vertex of type $i$ is matched to offline vertex $j$.

\begin{theorem}[c.f., \citet{HuangS:STOC:2021}]
    \label{thm:natural-lp-poly-time}
    The natural LP can be solved in polynomial time.
\end{theorem}

\begin{theorem}[Converse Jensen Inequality, c.f., \citet{HuangS:STOC:2021}]
    \label{thm:converse-jensen}
    For any convex $f : [0, 1] \to \R$ such that $f(0) = 0$, and any offline vertex $j \in J$ with $x_j = \sum_i x_{ij}$:
    \[
        \sum_{i \in I} \lambda_i f\Big( \frac{x_{ij}}{\lambda_i} \Big) \le \int_0^{-\ln(1-x_j)} f\big(e^{-\lambda}\big) d\lambda
        ~.
    \]
\end{theorem}

Denote function $\max \{x, 0\}$ as $x^+$.
Considering $f(x) = \big(x - \frac{1}{2}\big)^+$ gives the next corollary.

\begin{corollary}
    \label{cor:inverse-jensen}
    For any offline vertex $j \in J$, $\sum_{i \in I} \big( x_{ij} - \frac{\lambda_i}{2} \big)^+ \le \frac{1 - \ln 2}{2}$.
    %:
    %
    %\[
    %    \sum_{i \in I} \Big( x_{ij} - \frac{\lambda_i}{2} \Big)^+ \le \frac{1 - \ln 2}{2}
    %    ~.
    %\]
    %
\end{corollary}

\section{Edge-weighted Matching with Free Disposal}
\label{sec:edge-weighted-differential-inequality}

\subsection{Top Half Sampling}
\label{sec:top-half-sample}

%The suggested matching algorithm proposed by \citet{FeldmanMMM:FOCS:2009} matches each edge with a probability corresponding to the LP solution, achieving an $1-\frac{1}{e}$ competitive ratio. Our basic idea is to increase the probability of matching each edge, to reduce waste when few unmatched offline vertices remain. Specifically, our algorithm tries to match each edge with a probability doubling the LP solution. We call it Top Half Samping, since when it doubles the matching probabilities, only the edges whose weights are in the higher half of the corresponding online vertex types have the chance to be matched.

Suppose that an online vertex of type $i$ arrives at time $0 \le t \le 1$.
%We first define the \emph{marginal weight} of matching to each offline neighbor $j$.
Let $w_j(t)$ denote the maximum edge-weight matched to $j$ right before time $t$.
The \emph{marginal weight} of edge $(i, j)$ is:
\[
    w_{ij}(t) \defeq \big( w_{ij} - w_j(t) \big)^+
    ~.
\]

Further let $\succ_{i,t}$ be the total order over offline neighbors in descending order of the marginal weights, breaking ties arbitrarily.
%For any $(i,j) \in E$, define $w_{ij}(t)$ to be the remaining weight of edge $(i,j)$:
%\[
%    w_{ij}(t)=\max\big\{ w_{ij}-\max_{i'\in I_j: i' \text{ has matched to } j \text{ before time } t} \{w_{i'j}\} ,0 \big\} ~.
%\]
Hence, if $j \ne j' \in J_i$ and $j \succ_{i,t} j'$ then $w_{ij}(t) \ge w_{ij'}(t)$.
Naturally $j \succeq_{i,t} j'$ means that either $j \succ_{i,t} j'$ or $j = j'$.
%For any $i \in I$ and any $j,j' \in J_i$, define $j <_{i,t} j'$ if $w_{ij}(t) < w_{ij'}(t)$ or $w_{ij}(t) = w_{ij'}(t)$ and $j < j'$.

Finally, define $\sigma_{i,t} : [0, \lambda_i) \to J_i$ such that for any $\theta \in [0,\lambda_i)$, $\sigma_{i,t}(\theta)$ satisfies:
\[
    \sum_{j\in J_i: j \succ_{i,t} \sigma_{i,t}(\theta)} x_{ij} ~ \leq ~ \theta ~ < ~ \sum_{j\in J_i: j \succeq_{i,t} \sigma_{i,t}(\theta)} x_{ij}
    ~.
\]
If $\theta \geq \sum_{j\in J_i} x_{ij}$, artificially define $\sigma_{i,t}(\theta) =\,\perp$ as a dummy with zero marginal weight.
Intuitively, $\sigma_{i,t}(\theta)$ is the $\theta$-th heaviest neighbor of online vertex type $i$ at time $t$ where each edge occupies $x_{ij}$ ``slots''.
The algorithm matches the online vertex to $j = \sigma_{i,t}(\theta)$ for $\theta$ sampled uniformly from $[0, \frac{\lambda_i}{2})$.
See Figure~\ref{fig:top-half-sample} for an illustration.

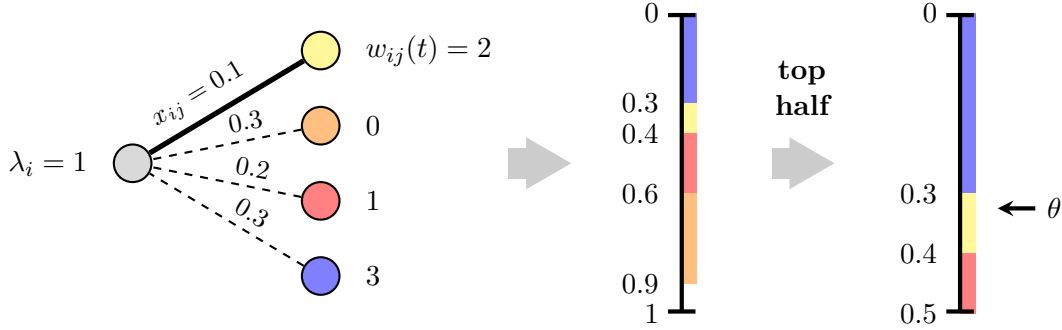
\begin{figure}[t]
    \centering
    \begin{tikzpicture}
        \tikzset{vertex/.style={draw=black,thick,fill=gray!30,circle,minimum width=.5cm}}
        \tikzset{edge/.style={draw=black,thick,dashed}}
        \begin{scope}
            \node[vertex,fill=yellow!50] (1) at (0,0) {};
            \node[right=.2cm of 1] {$w_{ij}(t)=2$};
            \node[vertex,fill=orange!50] (2) at (0,-1) {};
            \node[right=.2cm of 2] {$0$};
            \node[vertex,fill=red!50]    (3) at (0,-2) {};
            \node[right=.2cm of 3] {$1$};
            \node[vertex,fill=blue!50]   (4) at (0,-3) {};
            \node[right=.2cm of 4] {$3$};
            \node[vertex] (o) at (-2.5,-1.5) {};
            \node[left=.2cm of o] {$\lambda_i=1$};
            \path[edge,solid,line width=2pt] (1) edge node[pos=0.6,sloped,above] {\small $x_{ij} = 0.1$} (o);
            \path[edge] (2) edge node[pos=0.35,sloped,above] {\small $0.3$} (o);
            \path[edge] (3) edge node[pos=0.35,sloped,above] {\small $0.2$} (o);
            \path[edge] (4) edge node[pos=0.4,sloped,above] {\small $0.3$} (o);
        \end{scope}
        \begin{scope}[shift={(2.5,-1.5)}]
            \draw[draw=none,fill=gray!40] (0,-.2) rectangle +(.4,.4);
            \node[isosceles triangle,isosceles triangle apex angle=60,draw=none,fill=gray!40,minimum size=.6cm] at (.4,0) {};
        \end{scope}
        \begin{scope}[shift={(4.8,0)}]
            \draw[draw=none,fill=blue!50] (0,.5) rectangle +(.2,-1.2);
            \draw[draw=none,fill=yellow!50] (0,-0.7) rectangle +(.2,-0.4);
            \draw[draw=none,fill=red!50] (0,-1.1) rectangle +(.2,-0.8);
            \draw[draw=none,fill=orange!50] (0,.-1.9) rectangle +(.2,-1.2);
            \draw[|-|,black,ultra thick] (0,.5)--(0,-3.5);
            \node at (-0.4,0.5) {$0$};
            \node at (-0.55,-0.7) {$0.3$};
            \node at (-0.55,-1.1) {$0.4$};
            \node at (-0.55,-1.9) {$0.6$};
            \node at (-0.55,-3.1) {$0.9$};
            \node at (-0.4,-3.5) {$1$};
        \end{scope}
        \begin{scope}[shift={(6,-1.5)}]
            \node[align=center] at (.4,1) {\bf top\\\bf half};
            \draw[draw=none,fill=gray!40] (0,-.2) rectangle +(.4,.4);
            \node[isosceles triangle,isosceles triangle apex angle=60,draw=none,fill=gray!40,minimum size=.6cm] at (.4,0) {};
        \end{scope}
        \begin{scope}[shift={(8.5,0)}]
            \draw[draw=none,fill=blue!50] (0,.5) rectangle +(.2,-2.4);
            \draw[draw=none,fill=yellow!50] (0,-1.9) rectangle +(.2,-0.8);
            \draw[draw=none,fill=red!50] (0,-2.7) rectangle +(.2,-0.8);
            \draw[|-|,black,ultra thick] (0,.5)--(0,-3.5);
            \node at (-0.4,0.5) {$0$};
            \node at (-0.55,-1.9) {$0.3$};
            \node at (-0.55,-2.7) {$0.4$};
            \node at (-0.55,-3.5) {$0.5$};
            \draw[-stealth,ultra thick] (1,-2.1)--(.5,-2.1) node[right=.5cm] {$\theta$};
        \end{scope}
    \end{tikzpicture}
    \caption{Illustration of how Top Half Sampling matches an online vertex of type $i$ arriving at $t$.}
    \label{fig:top-half-sample}
\end{figure}

We may interpret the $1-\frac{1}{e}$ competitive Suggested Matching algorithm of \citet{FeldmanMMM:FOCS:2009} as sampling $\theta$ uniformly from $[0, \lambda_i)$, and defining $\sigma_{i,t}$ by an arbitrary order of offline neighbors, not necessarily the descending order of marginal weights.
Our algorithm, by contrast, samples from the top half of offline neighbors by the marginal weights.
We therefore call it Top Half Sampling.

\begin{tcolorbox}[beforeafter skip=10pt]%[title=Poisson Online Correlated Selection, boxsep=6pt, fonttitle=\bfseries, beforeafter skip=10pt]
    \textbf{Top Half Sampling}\\[1ex]
    \emph{Input at the beginning:}
    \begin{itemize}[itemsep=0pt, topsep=4pt]
        \item Online types $I$, offline vertices $J$, edges $E$;
        \item Arrival rates $(\lambda_i)_{i \in I}$, edge weights $(w_{ij})_{(i,j) \in E}$;
        \item Fractional matching $(x_{ij})_{(i,j) \in E}$ that satisfies Corollary~\ref{cor:inverse-jensen}, e.g., the Natural LP solution.
    \end{itemize}
    \smallskip
    \emph{When an online vertex of type $i \in I$ arrives at time $0 \le t \le 1$:}
    \begin{itemize}[itemsep=0pt, topsep=4pt]
        \item Sample $\theta$ uniformly from $[0,\frac{\lambda_i}{2})$;
        \item Match the online vertex to $j=\sigma_{i,t}(\theta)$.
    \end{itemize}
\end{tcolorbox}

\subsection{Analysis}

%Next we prove the main result of the section.

\begin{theorem}
    \label{thm:top-half-sample}
    Top Half Sampling is $\Gamma$-competitive for edge-weighted online stochastic matching with free disposal where:
    \[
        \Gamma = 1 - \frac{1}{1-\ln 2} \Big( \frac{1}{2e} - \frac{\ln 2}{e^2} \Big) > 0.7062
        ~.
    \]
\end{theorem}

Recall that $\opt$ is the optimal LP objective.
For any time $t$, let $\alg(t)$ denote the objective of the algorithm right before time $t$, and further define $\bar{\alg}(t) = \opt - \alg(t)$.
Our main lemma is the next differential inequality.
We next prove Theorem~\ref{thm:top-half-sample} assuming the lemma's correctness, deferring the lemma's proof to the end of section.

\begin{lemma}
    \label{lem:diff-eq}
    For any $0 \le t \le 1$:
    \begin{equation}
        \label{eqn:amplify-sample-lde}
        (2 + 2 \ln 2)\,\E\,\bar{\alg}(t) + (3 + \ln 2)\,\frac{d}{dt}\,\E\,\bar{\alg}(t) + \frac{d^2}{dt^2}\,\E\,\bar{\alg}(t) \le 0
        ~.
    \end{equation}
\end{lemma}

The following boundary conditions follow from the definitions of $\bar{\alg}(t)$ and Top Half Sampling. For completeness, Appendix~\ref{app:alg-init-rate} includes the proof of Lemma~\ref{lem:alg-init-rate}.

\begin{lemma}
    \label{lem:alg-init-val}
    $\bar{\alg}(0) = \opt$.
\end{lemma}

\begin{lemma}
    \label{lem:alg-init-rate}
    $\frac{d}{dt} \E \bar{\alg}(t=0) \le - \opt$.
\end{lemma}

\begin{proof}[Proof of Theorem~\ref{thm:top-half-sample}]
    By the differential inequality and boundary conditions in Lemmas \ref{lem:diff-eq}, \ref{lem:alg-init-val} and \ref{lem:alg-init-rate}:
    \[
        \frac{\E \bar{\alg}(1)}{\opt} \le \frac{1}{1-\ln 2} \Big( \frac{1}{2e} - \frac{\ln 2}{e^2} \Big)
        ~,
    \]
    with equality achieved when for any $0 \le t \le 1$:
    \[
        \frac{\E \bar{\alg}(t)}{\opt} = \frac{1}{1-\ln 2} \Big( \frac{1}{(2e)^t} - \frac{\ln 2}{e^{2t}} \Big)
        ~.
    \]

    For completeness, Appendix~\ref{app:edge-weighted} includes this standard analysis of differential inequality.
    Now the theorem follows since the expected objective of Top Half Sampling is $\E\,\alg(1) = \opt - \E\,\bar{\alg}(1)$.
\end{proof}

\subsubsection{Proof of Lemma~\ref{lem:diff-eq}: Unweighted Case}
\label{sec:ths-analysis-unweighted}

It is instructive to first consider unweighted matching.
%and to develop some technical ingredients.
The unweighted assumption is only used to relate the differential inequality with our variables;
\emph{all inequalities and lemmas about the variables themselves continue to hold in the edge-weighted case.}
Generalizing to the edge-weighted case will be relatively straightforward given these ingredients.
%We start by analyze our algorithm in unweighted matching, where we only care about the number of matched edges. We first introduce some further notations (which are consistent and also used in edge-weighted matching).

For any time $t$ and any edge $(i, j)$, let $x_{ij}(t) = x_{ij}$ if $j$ is unmatched right before time $t$, and $0$ otherwise.
Let $x_i(t) = \sum_{j \in J_i} x_{ij}(t)$ for any online vertex type $i$.
%Define:
%
%\[
%    \rho_{ij}(t) \defeq \frac{x_{ij}(t)}{\lambda_i}
%    ~,\quad
%    \rho_i(t) \defeq \frac{x_i(t)}{\lambda_i}
%    ~.
%\]
%
%They are essentially the LP's suggested probabilities matching $i$ to $j$ at time $t$, and the probability of matching $i$.
%Suggested Matching \cite{FeldmanMMM:FOCS:2009}, for example, would match $i$ to $j$ with probability $\rho_{ij}(t)$.
%
%the probability of matching an online vertex $i$ arriving at time $t$ to $j$, and $\rho_i(t) = \frac{x_i(t)}{\lambda_i}$, the probability of an online vertex $i$ arriving at time $t$.
%We further introduce some notations about the match probabilities of Top Half Sampling.
Suppose that an online vertex of type $i$ arrives at time $t$.
Consider function $p(\rho) = \min \{ 2 \rho, 1\}$.
By definition, Top Half Sampling at most doubles the LP match probability, and therefore matches this online vertex with probability $p(\frac{x_i(t)}{\lambda_i})$.
Since type $i$ vertices arrive at rate $\lambda_i$, they contribute $\lambda_i \, p(\frac{x_i(t)}{\lambda_i})$ to the algorithm's match rate (i.e., $-\frac{d}{dt}\bar{\alg}(t)$).
Further let $y_{ij}(t)$ be edge $(i, j)$'s contribution.
We have:
%Top Half Sampling matches it to each offline neighbor $j$.
%
\begin{equation}
    \label{eqn:amplify-sample-y-1}
    \sum_{j \in J_i} y_{ij}(t) = \lambda_i\,p \Big( \frac{x_i(t)}{\lambda_i} \Big)
    ~.
\end{equation}

The definition of Top Half Sampling further implies:
\begin{equation}
    \label{eqn:amplify-sample-y-2}
    \underbrace{\vphantom{\bigg[}\lambda_i\,p\Big(\frac{x_i(t)}{\lambda_i}\Big)}_\text{match rate of $i$} - \underbrace{\vphantom{\bigg[}\lambda_i\,p\Big(\frac{x_i(t)}{\lambda_i}-\frac{x_{ij}(t)}{\lambda_i}\Big)}_\text{max match rate of $(i,j')$ for $j' \ne j$} \le \underbrace{\vphantom{\bigg[}~y_{ij}(t)~}_\text{match rate of $(i,j)$} \le \underbrace{\vphantom{\bigg[}\min \big\{ 2 x_{ij}(t), \lambda_i \big\}}_\text{doubled LP match rate, capped by $\lambda_i$}
    %~.
    %
\end{equation}

%For any $(i,j)\in E$, suppose the probability of matching $i$ to $j$ at time $t$ is $\frac{y_{ij}(t)}{\lambda_i}$. Then $\big(y_{ij}(t) \big)_{(i,j) \in E}$ satisfies that:

%Eqn.~\eqref{eqn:amplify-sample-y-1} substantiates the aforementioned description that the top half sampling algorithm matches an online vertex of type $i$ with probability $p\big(\frac{x_i(t)}{\lambda_i}\big)$.
%The first inequality follows because vertices other than $j$ contribute at most $p(\rho_i(t)-\rho_{ij}(t))$ to $i$'s match probability at time $t$;
%the remaining must come from $j$.
%The second follows because the probability of matching $(i, j)$ at most doubles the natural LP solution.

We next prove the main lemma about function $p$.

\begin{lemma}
    \label{lem:amplify-sample-p}
    For any $0 \le \rho_i \le 1$ and any $\rho_{ij} \ge 0$ such that $\sum_{j \in J_i} \rho_{ij} = \rho_i$:
    \[
        p(\rho_i) - \rho_i \ge \frac{1}{2} \sum_{j \in J_i} \Big( p(\rho_i) - p(\rho_i - \rho_{ij}) - (2\rho_{ij}-1)^+ \Big)
        ~.
    \]
\end{lemma}

\begin{proof}
    If $\rho_i \le \frac{1}{2}$, we have $p(\rho) = 2\rho$ on both sides, and $(2\rho_{ij} - 1)^+ = 0$ for all $j$.
    Both sides are $\rho_i$.

    Next suppose that $\rho_i > \frac{1}{2}$.
    The left-hand-side then equals $1 - \rho_i$.

    If $\rho_{ij} > \frac{1}{2}$ for some $j$, for all other $j' \ne j$ we have $\rho_i - \rho_{ij'} > \frac{1}{2}$ and $p(\rho_i) = p(\rho_i - \rho_{ij'}) = 1$.
    Hence, the right-hand-side equals $\frac{1}{2} \big( 1 - 2(\rho_i - \rho_{ij}) - (2\rho_{ij}-1) \big) = 1 - \rho_i$.
    
    Otherwise, the right-hand-side is $\sum_j \big( 1 - p(\rho_i - \rho_{ij}) \big)$.
    Since $1 - p(\rho_i - \rho_{ij})$ is convex in $\rho_{ij}$, it is maximized when there are $j_1, j_2$ with $\rho_{ij_1} = \frac{1}{2}$ and $\rho_{ij_2} = \rho_i - \frac{1}{2}$, with maximum value $1 - \rho_i$.
\end{proof}

\paragraph{Proof of Eqn.~\eqref{eqn:amplify-sample-lde}.}
%
%We denote $y_{ij}(t) = \lambda_i \mu_{ij}(t)$ as the match rate of edge $(i, j)$ at time $t$, and $y_j(t) = \sum_i y_{ij}(t)$ as the total match rate of offline vertex $j$.
We shall prove it for any time $t$ even conditioned on the randomness before time $t$ both in the Poisson arrivals and in the algorithm.
The remaining argument omits $t$ in the variables for notational simplicity.
Since $y_{ij}$ is the match rate of edge $(i, j)$, let $y_j = \sum_i y_{ij}$ be the total match rate of offline vertex $j$.
Next we write the three terms of the differential inequality by our notations.
The zeroth-order term satisfies:
\[
    \bar{\alg} \le \sum_{i \in I} x_i = \sum_{(i,j) \in E} x_{ij}
    ~.
\]
It holds with equality at time $0$.
Further, the right-hand-side decreases by at most $1$ when an offline vertex is matched, and the left decreases by exactly $1$.
The first-order and second-order terms are:
\begin{align*}
    \frac{d}{dt} \bar{\alg} & ~=~ - \, \sum_{j \in J} y_j
    \\
    \frac{d^2}{dt^2} \bar{\alg} & ~=~ \sum_{j \in J} 
    \underbrace{\vphantom{\bigg[}y_j}_\text{match rate of $j$}
        \sum_{i \in I_j} ~ \bigg( \underbrace{\vphantom{\bigg[}\lambda_i \, p \Big(\frac{x_i}{\lambda_i}\Big)}_\text{current match rate of $i$} - \underbrace{\vphantom{\bigg[}\lambda_i \, p \Big(\frac{x_i}{\lambda_i} - \frac{x_{ij}}{\lambda_i}\Big) }_\text{match rate of $i$ after matching $j$} \bigg)
    ~.
\end{align*}

Writing $\Delta_{ij} = (2x_{ij}-\lambda_i)^+$ and $\Delta_j = \sum_{i \in I_j} \Delta_{ij}$,
Equation~\eqref{eqn:amplify-sample-y-2} implies that:
\begin{equation}
    \label{eqn:edge-weighted-y-upper-bound}
    y_j
    \le \sum_{i \in I_j} \min \big\{ 2 x_{ij}, \lambda_i \big\}
    = 2 \sum_{i \in I_j} x_{ij} - \sum_{i \in I_j} \Delta_{ij}
    \le 2 - \Delta_j
    ~.
\end{equation}

We can therefore relax the second-order term as follows.
This relaxation is superfluous in unweighted matching but will be necessary in the edge-weighted case.
\begin{align*}
    \frac{d^2}{dt^2} \bar{\alg}
    &
    = \sum_{j \in J} y_j \sum_{i \in I_j} \bigg( \lambda_i \, p \Big(\frac{x_i}{\lambda_i}\Big) - \lambda_i \, p \Big(\frac{x_i}{\lambda_i}-\frac{x_{ij}}{\lambda_i}\Big) - \Delta_{ij} \bigg) + \sum_{j \in J} y_j \Delta_j
    \\
    &
    \le \sum_{j \in J} \big(2 - \Delta_j \big) \sum_{i \in I_j} \bigg( \lambda_i \, p \Big(\frac{x_i}{\lambda_i}\Big) - \lambda_i \, p \Big(\frac{x_i}{\lambda_i}-\frac{x_{ij}}{\lambda_i}\Big) - \Delta_{ij} \bigg) + \sum_{j \in J} y_j \Delta_j
    ~.
\end{align*}

Finally, Corollary~\ref{cor:inverse-jensen} can be restated as for any $j \in J$:
\begin{equation}
    \label{eqn:amplify-sample-delta}
    \Delta_j \le 1 - \ln 2
    ~.
\end{equation}

Replacing the three terms of differential inequality~\eqref{eqn:amplify-sample-lde} by the above relaxations and rearranging terms, it reduces to the following lemma.
%$y_{ij}(t) \le \lambda_i$ holds trivially since the matching probability $\frac{y_{ij}(t)}{\lambda_i}$ cannot exceed $1$.

%We remark that any $\big(y_{ij}(t)\big)_{(i,j)\in E}$ that satisfies Eqn.~\eqref{eqn:amplify-sample-y-1} and the second inequality in Eqn.~\eqref{eqn:amplify-sample-y-2} would also satisfy the first.

%Finally, define $\Delta_j(t) = \sum_{i \in I_j} \lambda_i ( 2\rho_{ij}(t) - 1 )^+$.

%Finally Eqn.~\eqref{eqn:amplify-sample-y-2} implies the next lemma, which is what we truly need in the analysis.

%\begin{lemma}
%    \label{lem:amplify-sample-y}
    %
%    For any time $0\leq t\leq 1$ and any $j \in J$ we have:
    %
%    \[
%        y_j(t) \le 2 - \Delta_j(t)
%        ~.
%    \]
    %
%\end{lemma}

%

%\begin{proof}
%    We shall prove the Inequality~\eqref{eqn:amplify-sample-lde} at any time $0 \le t \le 1$ conditioned on any realization of the remaining graph. Bounding the three terms with our notations:
    %

%    Plugging them into Equation ~\eqref{eqn:amplify-sample-lde}, omitting $t$ for notational simplicity, and rearranging terms, it suffices to prove the next lemma.
%\end{proof}

\begin{lemma}
    \label{lem:diff-eq-2}
    Suppose that nonnegative $(x_{ij})_{(i,j)\in E}$, $(y_{ij})_{(i,j)\in E}$ and $(\Delta_j)_{j \in J}$ satisfy Equations \eqref{eqn:amplify-sample-y-1}, \eqref{eqn:amplify-sample-y-2}, and \eqref{eqn:amplify-sample-delta}.
    Further suppose that $x_i = \sum_j x_{ij} \le \lambda_i$, $y_j = \sum_i y_{ij}$, and $\Delta_{ij} = (2 x_{ij} - \lambda_i)^+$.%
    \footnote{This lemma does not assume $\sum_i \Delta_{ij} = \Delta_j$, an important flexibility that will be useful in the edge-weighted case.}
    Then:
    \[
        \sum_{j \in J} \big( 2 - \Delta_j \big) \sum_{i \in I_j} \bigg( \lambda_i\,p\Big(\frac{x_i}{\lambda_i}\Big) - \lambda_i\,p\Big(\frac{x_i}{\lambda_i}-\frac{x_{ij}}{\lambda_i}\Big) - \Delta_{ij} \bigg) + \sum_{j \in J} y_j \big( \Delta_j - 1 + \ln 2 \big) \le (2+2\ln 2) \sum_{(i,j) \in E} \big( y_{ij} - x_{ij} \big)
        ~.
    \]
\end{lemma}

\begin{proof}
    Consider any $\big( z_{ij} \big)_{(i,j) \in E}$ that satisfies $\sum_j z_{ij} = 2 \sum_j ( y_{ij} - x_{ij} )$ for all $i \in I$ and:
    \begin{equation}
        \label{eqn:amplify-sample-z}
        \lambda_i\,p\Big(\frac{x_i}{\lambda_i}\Big) - \lambda_i\,p\Big(\frac{x_i}{\lambda_i}-\frac{x_{ij}}{\lambda_i}\Big) - \Delta_{ij} \le z_{ij} \le y_{ij}
        ~.
    \end{equation}
    It exists because (i) by the first part of Eqn.~\eqref{eqn:amplify-sample-y-2}, the lower bound of $z_{ij}$ is smaller than or equal to the upper bound; (ii) by Lemma~\ref{lem:amplify-sample-p} and Eqn.~\eqref{eqn:amplify-sample-y-1} the lower bound sums to at most $2 \sum_j ( y_{ij} - x_{ij} )$; and (iii) by the second part of Eqn.~\eqref{eqn:amplify-sample-y-2} the upper bound sums to at least $2 \sum_j ( y_{ij} - x_{ij} )$.
    
    It remains to prove that:
    \[
        \sum_j \big( 2 - \Delta_j \big) \sum_i z_{ij} + \sum_j y_j \big( \Delta_j - 1 + \ln 2 \big) \le \big(1+\ln 2\big) \sum_j \sum_i z_{ij}
        ~.
    \]

    Rearrange terms and this becomes:
    \[
        \sum_j \big(1 - \ln 2 - \Delta_j \big) \sum_i \big(y_{ij} - z_{ij}\big) \ge 0
        ~,
    \]
    which follows by Equations~\eqref{eqn:amplify-sample-delta} and \eqref{eqn:amplify-sample-z}.
\end{proof}

\subsubsection{Proof of Lemma~\ref{lem:diff-eq}: Edge-weighted Case}\label{sec:ths-analysis-edge-weighted}

%Now we analyze the top half sampling algorithm in edge-weighted matching with free disposals. First notice that Lemmas $\ref{lem:amplify-sample-y}$, $\ref{lem:alg-init-val}$ and $\ref{lem:alg-init-rate}$ still hold since they do not relate to anything specifically about unweighted matching. For the same competitive ratio $0.7062$, we only need to prove the differential inequality \eqref{eqn:amplify-sample-lde}, since the initial conditions still hold. Broadly speaking, we will write the three terms in Equation \eqref{eqn:amplify-sample-lde} as integrals on a new dimension representing the weights to complete the proof.
Like the unweighted case, we shall prove differential inequality~\eqref{eqn:amplify-sample-lde} at any time $t$ conditioned on any realization of randomness before time $t$ both in the Poisson arrivals and in the algorithm.
We will drop $t$ from all variables except for the marginal weights for notational simplicity.

%We next introduce some notations to better express the three terms of the differential inequality.
For any edge $(i,j)$ and any weight-level $w \ge 0$, define $x_{ij}(w) = x_{ij}$ if the marginal weight meets the weight-level, i.e., $w_{ij}(t) \ge w$, and $x_{ij}(w) = 0$ otherwise.
Let $x_i(w) = \sum_j x_{ij}(w)$.
%Further define $\rho_{ij}(v,t)=\frac{x_{ij}(v,t)}{\lambda_i}$ and $\rho_i(v,t)=\frac{x_i(v,t)}{\lambda_i}$.
Recall that $y_{ij}$ is the match rate of edge $(i, j)$, i.e., the product of type $i$'s arrival rate $\lambda_i$ and the probability that Top Half Sampling would match a type $i$ online vertex to $j$ should it arrive at time $t$; and $y_j = \sum_i y_{ij}$.
For any weight-level $w \ge 0$ let $y_{ij}(w) = y_{ij}$ if $w_{ij}(t) \ge w$ and $y_{ij}(w) = 0$ otherwise.
Finally, let $y_j(w)=\sum_i y_{ij}(w)$.

The zeroth-order term can be bounded by:
\[
    \bar{\alg} \le \sum_{(i,j) \in E} x_{ij} w_{ij}(t) = \int_{0}^\infty \sum_{(i,j) \in E} x_{ij}(w)\,dw
    ~.
\]
The inequality holds with equality at time $0$.
Further, when an edge is matched the left decreases by exactly the marginal weight, while the right decreases by at most the marginal weight. 

The first-order term is:
\[
    \frac{d}{dt} \bar{\alg} = - \sum_{(i,j) \in E} y_{ij} w_{ij}(t) = - \int_0^\infty \sum_{(i,j) \in E} y_{ij}(w) dw
    ~.
\]

Finally consider the second-order term.
The match rate of edge $(i,j)$ is $y_{ij}$.
If such an edge is matched, the marginal weight of any edge $(i',j)$ decreases from $w_{i'j}(t)$ to $( w_{i'j}(t) - w_{ij}(t) )^+$.
Hence:
\[
    \frac{d^2}{dt^2} \bar{\alg} = \sum_{(i,j)} y_{ij} \sum_{i'} \int_{(w_{i'j}(t)-w_{ij}(t))^+}^{w_{i'j}(t)} \bigg( \lambda_{i'} \, p \Big( \frac{x_{i'}(w)}{\lambda_{i'}}\Big) - \lambda_{i'} \, p \Big( \frac{x_{i'}(w)}{\lambda_{i'}} - \frac{x_{i'j}(w)}{\lambda_{i'}} \Big) \bigg) dw
    ~.
\]

Recall that $\Delta_{ij} = \big( 2 x_{ij} - \lambda_i \big)^+$, $\Delta_j = \sum_i \Delta_{ij}$, and $y_j \le 2 - \Delta_j$ (i.e., Eqn.~\eqref{eqn:edge-weighted-y-upper-bound}).
Further for any weight-level $w \ge 0$ let $\Delta_{ij}(w) = \big( 2 x_{ij}(w) - \lambda_i \big)^+ \le \Delta_{ij}$.
%and $\Delta_j(w,t) = \sum_i \Delta_{ij}(w,t)$.
The definition of function $p$ implies:
\begin{equation}
    \label{eqn:edge-weighted-2nd-order-part1}
    \lambda_i \, p \Big( \frac{x_{i}(w)}{\lambda_i}\Big) - \lambda_i \, p \Big( \frac{x_{i}(w)}{\lambda_i} - \frac{x_{ij}(w)}{\lambda_i} \Big) - \Delta_{ij}(w) \ge 0
    ~.
\end{equation}

We bound the second-order term in two parts.
This corresponds to the superfluous relaxation in the unweighted case, although it is necessary and more involved this time.
The two parts are:
\begin{align*}
    &
    \sum_{(i,j)} y_{ij} \sum_{i'} \int_{(w_{i'j}(t)-w_{ij}(t))^+}^{w_{i'j}(t)} \bigg( \lambda_{i'} \, p \Big( \frac{x_{i'}(w)}{\lambda_{i'}}\Big) - \lambda_{i'} \, p \Big( \frac{x_{i'}(w)}{\lambda_{i'}} - \frac{x_{i'j}(w)}{\lambda_{i'}} \Big) - \Delta_{i'j}(w) \bigg) dw
    \\
    & \quad \le
    \int_0^\infty \sum_{(i,j)} y_{ij} \sum_{i'} \bigg( \lambda_{i'} \, p \Big( \frac{x_{i'}(w)}{\lambda_{i'}}\Big) - \lambda_{i'} \, p \Big( \frac{x_{i'}(w)}{\lambda_{i'}} - \frac{x_{i'j}(w)}{\lambda_{i'}} \Big) - \Delta_{i'j}(w) \bigg) dw
    \tag{Eqn.~\eqref{eqn:edge-weighted-2nd-order-part1}}
    \\
    & \quad \le
    \int_0^\infty \sum_j \big( 2 - \Delta_j \big) \sum_{i'} \bigg( \lambda_{i'} \, p \Big( \frac{x_{i'}(w)}{\lambda_{i'}}\Big) - \lambda_{i'} \, p \Big( \frac{x_{i'}(w)}{\lambda_{i'}} - \frac{x_{i'j}(w)}{\lambda_{i'}} \Big) - \Delta_{i'j}(w) \bigg) dw
    ~,
    \tag{Eqn.~\eqref{eqn:edge-weighted-y-upper-bound}}
\end{align*}
and
\begin{align*}
    \sum_{(i,j)} y_{ij} \sum_{i'} \int_{(w_{i'j}(t)-w_{ij}(t))^+}^{w_{i'j}(t)} \Delta_{i'j}(w) \, dw
    %& ~ \le ~
    %\sum_{(i,j)} y_{ij} \sum_{i'} \int_{(w_{i'j}(t)-w_{ij}(t))^+}^{w_{i'j}(t)} \Delta_{i'j} \, dw
    %\\[.5ex]
    & ~ \le ~
    \sum_{(i,j)} y_{ij} \sum_{i'} w_{ij}(t) \Delta_{i'j}
    \tag{$\Delta_{i'j}(w) \le \Delta_{i'j}$}
    \\
    & ~ = ~
    \sum_{(i,j)} y_{ij} w_{ij}(t) \Delta_j
    ~ = ~
    \int_0^\infty \sum_{(i,j)} y_{ij}(w) \Delta_j \, dw
    ~.
\end{align*}

    Together we have that $\frac{d^2}{dt^2} \bar{\alg}$ is at most:
    \[
        \int_0^\infty \sum_j \bigg( \big( 2 - \Delta_j \big) \ \sum_i \Big( \lambda_i \, p \Big( \frac{x_i(w)}{\lambda_i}\Big) - \lambda_i \, p \Big( \frac{x_i(w)}{\lambda_i} - \frac{x_{ij}(w)}{\lambda_i} \Big) - \Delta_{ij}(w) \Big) + \sum_{(i,j)} y_{ij}(w) \Delta_j \bigg) dw
        ~.
    \]

    Now that we have written all three terms as integrals over weight-levels from $0$ to $\infty$, it suffices to prove that for any weight-level $w \ge 0$ (rearranging terms):
    \begin{align*}
        &
        \sum_j \big( 2 - \Delta_j \big) \sum_i \bigg( \lambda_i\,p\Big(\frac{x_i(w)}{\lambda_i}\Big) - \lambda_i\,p\Big(\frac{x_i(w)}{\lambda_i}-\frac{x_{ij}(w)}{\lambda_i}\Big) - \Delta_{ij}(w) \bigg) + \sum_j y_j(w) \big( \Delta_j - 1 + \ln 2 \big) \\
        & \hspace{.6\textwidth}
        \le (2+2\ln 2) \sum_{(i,j)} \big( y_{ij}(w) - x_{ij}(w) \big)
        ~.
    \end{align*}

    This follows by Lemma~\ref{lem:diff-eq-2}, as it is easy to verify that the variables satisfy the lemma's conditions.

\section{Vertex-weighted Matching}
\label{sec:vertex-weighted-ocs}

\subsection{Poisson Online Correlated Selection}
\label{sec:poisson-ocs}

In vertex-weighted and unweighted matching, we only need to consider the probability of matching each offline vertex.
Our algorithm is an online rounding in the Poisson arrival model.
It receives at the beginning a fractional matching $(x_{ij})_{(i,j) \in E}$ from some polytope, e.g., that of the Natural LP.
Then, when an online vertex of type $i$ arrives, the algorithm matches it to an unmatched neighbor $j$ with positive $x_{ij}$.
For any offline vertex $j$ that is matched to $x_j = \sum_i x_{ij}$ in the fractional matching, our algorithm shall match $j$ with probability at least $\Gamma x_j$ by the end, where $\Gamma$ is the competitive ratio for the stochastic matching problem and its value depends on the tightness of the polytope.

Our algorithm and its analysis are inspired by the online correlated selection (OCS) literature, in particular, by the multi-way semi-OCS of \citet{GaoHHNYZ:FOCS:2021}.
Hence, we call it Poisson OCS.

\begin{tcolorbox}[beforeafter skip=10pt]%[title=Poisson Online Correlated Selection, boxsep=6pt, fonttitle=\bfseries, beforeafter skip=10pt]
    \textbf{Poisson Online Correlated Selection}\\[1ex]
    \emph{Input at the beginning:}
    \begin{itemize}[itemsep=0pt, topsep=4pt]
        \item Online types $I$, offline vertices $J$, edges $E$;
        \item Arrival rates $(\lambda_i)_{i \in I}$;
        \item Fractional matching $(x_{ij})_{(i,j) \in E}$ such that $\forall i \in I, \sum_j x_{ij} \le \lambda_i$, and $\forall j \in J, \sum_i x_{ij} \le 1$.
    \end{itemize}
    \smallskip
    \emph{Preprocessing:~}
    \begin{itemize}[itemsep=0pt, topsep=4pt]
        \item Compute $\rho_{ij} = \frac{x_{ij}}{\lambda_i}$ for any $(i, j) \in E$, and $x_j = \sum_i x_{ij}$ for any $j \in J$.
    \end{itemize}
    \smallskip
    \emph{When an online vertex of type $i \in I$ arrives at time $0 \le t \le 1$:}
    \begin{itemize}[itemsep=0pt, topsep=4pt]
        \item Match it to an unmatched offline neighbor $j$ with probability proportional to $e^{tx_j} \rho_{ij}$.
    \end{itemize}
\end{tcolorbox}

%Note that the algorithms in previous works \cite{JailletL:MOR:2014, HuangS:STOC:2021, HaeuplerMZ:WINE:2011, BrubachSSX:EDA:2016, ManshadiOS:MOR:2012} sample two (or three) choices for each online vertex. If these choices are already matched, the online vertex is discarded. On the other hand, our algorithm always matches an online vertex if there is any neighbor available.

\subsection{Meta Analysis}

For any offline vertex $j$ and any $0 \le t \le 1$, let $Y_j(t)$ be the indicator of whether $j$ stays unmatched from time $0$ to $t$.
Further let $\bar{Y}_j(t) = \E\,Y_j (t)$ be the probability that $j$ stays unmatched from time $0$ to $t$.
We seek to upper bound $\bar{Y}_j(1)$ for any $j \in J$.
More generally, for any subset $T \subseteq J$ of offline vertices and any $0 \le t \le 1$, we will consider $Y_T(t) = \prod_{j \in T} Y_j(t)$, and $\bar{Y}_T(t) = \E\,Y_T(t)$.

The proof of the next lemma is essentially the same as the counterpart in \citet{GaoHHNYZ:FOCS:2021}.
We include it nonetheless to be self-contained.

\begin{lemma}
    \label{lem:ocs-trick}
    For any subset $T \subseteq J$ of offline vertices and any time $0 \le t < 1$:
    \begin{equation}
        \label{eqn:ocs-trick}
        \frac{d}{dt} \bar{Y}_T(t) \le - \bar{Y}_T(t) \sum_i \lambda_i \frac{\sum_{j \in T} e^{tx_j} \rho_{ij} \bar{Y}_T(t)}{\sum_{j \in T} e^{tx_j} \rho_{ij} \bar{Y}_T(t) + \sum_{j \notin T} e^{tx_j} \rho_{ij} \bar{Y}_{T+j}(t)}
    \end{equation}
\end{lemma}

\begin{proof}
    For any time $0 \le t < 1$, any sufficiently small $\epsilon > 0$, any subset $T \subseteq J$ of offline vertices, and conditioned on any realization of $Y(t) = \big( Y_T(t) \big)_{T \subseteq J}$, we have:
    \[
        \E \big[ Y_T(t+\epsilon) \mid Y(t) \big] = Y_T(t) \bigg( 1 - \epsilon \sum_{i \in I} \lambda_i \frac{\sum_{j \in T} e^{t x_j} \rho_{ij} Y_j(t)}{\sum_{j \in J} e^{t x_j} \rho_{ij} Y_j(t)} \,\bigg) + O(\epsilon^2)
        ~.
    \]

    Multiplying the numerator and denominator by $Y_T(t)$, and using that $X^2 = X$ for $X \in \{0, 1\}$:
    \[
        \E \big[ Y_T(t+\epsilon) \mid Y(t) \big] = Y_T(t) \bigg( 1 - \epsilon \sum_{i \in I} \lambda_i \frac{\sum_{j \in T} e^{t x_j} \rho_{ij} Y_T(t)}{\sum_{j \in T} e^{t x_j} \rho_{ij} Y_T(t) + \sum_{j \notin T} e^{t x_j} \rho_{ij} Y_{T+j}(t)} \,\bigg) + O(\epsilon^2)
        ~.
    \]

    This is concave (jointly) in $Y_T(t)$ and $\sum_{j \notin T} e^{t x_j} \rho_{ij} Y_{T+j}(t)$ (with its value defined as $0$ when both terms are $0$).
    Next consider the expectation over the realization of $Y(t)$.
    By the Jensen Inequality:
    \[
        \bar{Y}_T(t+\epsilon) \le \bar{Y}_T(t) \bigg( 1 - \epsilon \sum_{i \in I} \lambda_i \frac{\sum_{j \in T} e^{t x_j} \rho_{ij} \bar{Y}_T(t)}{\sum_{j \in T} e^{t x_j} \rho_{ij} \bar{Y}_T(t) + \sum_{j \notin T} e^{t x_j} \rho_{ij} \bar{Y}_{T+j}(t)} \,\bigg) + O(\epsilon^2)
        ~.
    \]

    Taking the limit with $\epsilon$ going to $0$ proves the lemma.
\end{proof}

We can therefore upper bound $\bar{Y}_T(t)$ for any subset $T \subseteq J$ and any time $0 \le t \le 1$ by a family of functions satisfying differential inequality \eqref{eqn:ocs-trick} in the opposite direction.
The next lemma formalizes this.
We defer its proof to Appendix~\ref{app:meta-analysis}.

\begin{lemma}
    \label{lem:meta-analysis}
    If a family of absolutely continuous functions $p_T : [0, 1] \to \R^+$ for $T \subseteq J$ satisfy:
    \begin{align}
        p_T(0) & = 1 \mbox{ , and} \nonumber\\
        \frac{d}{dt} p_T(t) & \ge - \, p_T(t) \, \sum_i \lambda_i \, \frac{\sum_{j \in T} e^{tx_j} \rho_{ij} p_T(t)}{\sum_{j \in T} e^{tx_j} \rho_{ij} p_T(t) + \sum_{j \notin T} e^{tx_j} \rho_{ij} p_{T+j}(t)} \label{eqn:meta-analysis}
        \quad 
    \end{align}
    almost everywhere, then for any $T \subseteq J$ and any $0 \le t \le 1$:
    \[
        \bar{Y}_T(t) \le p_T(t)
        ~.
    \]
\end{lemma}

\paragraph{Rest of the Section in a Nutshell.}
If the fractional matching is unrestricted (other than being a fractional matching), we may have $\rho_{ij} = 1$ or $0$ for all $(i,j)$, and Eqn.~\eqref{eqn:meta-analysis} for each offline vertex $j$ degenerates to $\frac{d}{dt} p_j(t) \ge - p_j(t) x_j$, giving only a trivial bound $p_j(t) = e^{-tx_j}$ (Subsection~\ref{sec:0th-level}).
Using the constraints of the Natural LP and the Converse Jensen Inequality, we can already improve the state-of-the-art competitive ratio for vertex-weighted matching from $0.7$ to $0.707$ (Subsection~\ref{sec:1st-level}).
Further, we next introduce an LP hierarchy that generalizes the Natural LP (Subsection~\ref{sec:poisson-matching-lp}), and give a Converse Jensen Inequality for the second level LP of this hierarchy (Subsection~\ref{sec:2nd-converse-jensen}).
Using the second level LP we obtain our final competitive ratio $0.716$ (Theorem~\ref{thm:2nd-level}, Subsection~\ref{sec:2nd-level}).
We leave further improvements using higher level LPs as a future research direction.

\subsection{Poisson Matching Linear Program Hierarchy}
\label{sec:poisson-matching-lp}

%To further improve the competitive ratio, we need to consider a tighter LP than the Natural LP of \citet{HuangS:STOC:2021}.
%This subsection introduces a LP hierarchy for the Poisson arrival model.
%Our final analysis considers LP at the second level of the hierarchy, and we leave further improvements using higher level LPs as a future research direction.

Let $P_k(\lambda)$ be the value of Poisson cumulative distribution function with arrival rate $\lambda$ at $k \in \Z^+$:
\[
    P_k(\lambda) = e^{-\lambda} \sum_{\ell = 0}^k \frac{\lambda^\ell}{\ell!}
    ~.
\]

\begin{lemma}
    \label{lem:poisson-cdf-property}
    For any $k \in \Z^+$ and any $\lambda \ge 0$, 
    $\frac{d}{d\lambda} P_k(\lambda) = P_{k-1}(\lambda) - P_k(\lambda)$.
    %
    %\[
    %    ~.
    %\]
    %
\end{lemma}

\medskip

Its proof is basic calculus which we omit.
For any $\ell \ge 1$, the \emph{$\ell$-th level Poisson Matching LP} is:
\begin{align}
    \text{maximize} \quad
    &
    \sum_{(i,j) \in E} w_{ij} x_{ij} \notag \\%[1ex]
    \text{subject to} \quad
    &
    \sum_{j \in J} x_{ij} \le \lambda_i && \forall i \in I \notag \\
    &
    \sum_{i \in S} \sum_{j \in T} x_{ij} \le \sum_{1 \le k \le m} \big(1 - P_{k-1}(\lambda_S) \big) && \forall 1 \le m \le \ell, \forall S \subseteq I, \forall T \subseteq J : |T| = m \label{eqn:poisson-constraint} \\[.5ex]
    &
    x_{ij} \ge 0 && \forall (i, j) \in E \notag
\end{align}

Artificially let the matching LP (i.e., to maximize $\sum w_{ij} x_{ij}$ in the matching polytope which we will present as Eqn.~\eqref{eqn:matching-polytope}) be the $0$-th level Poisson Matching LP.
The matching constraints for offline vertices, i.e., $\sum_i x_{ij} \le 1$, are subsumed by constraints~\eqref{eqn:poisson-constraint} in the $\ell$-th level Poisson Matching LPs for any $\ell \ge 1$ with $S = I, T = \{j\}$.
The Natural LP is the first level Poisson Matching LP.

The next lemmas show that the LPs in the hierarchy are indeed relaxations of online stochastic matching, and are computationally tractable at the constant levels.
The proofs are similar to the counterparts for the Natural LP \cite{HuangS:STOC:2021} so we defer them to Appendices~\ref{app:poisson-matching-lp-relaxation} and \ref{app:poisson-matching-lp-poly-time}.

\begin{lemma}
    \label{lem:poisson-matching-lp-relaxation}
    For any online stochastic matching instance, the optimal objective of the instance is at most the optimal objective of the $\ell$-th level Poisson Matching LP at any level $\ell \ge 0$.
\end{lemma}

\begin{lemma}
    \label{lem:poisson-matching-lp-poly-time}
    For any $\ell \ge 0$, the $\ell$-th level Poisson Matching LP is solvable in time polynomial in $|I|$ and $|J|^{\max \{ \ell, 1 \}}$.
\end{lemma}

\subsection{Second Level Converse Jensen Inequality}
\label{sec:2nd-converse-jensen}

A smooth function $f : [0, 1]^d \to \R$ is \emph{diminishing returns (DR) submodular} if its Hessian matrix has no positive entries, i.e.:
\[
    \forall i, j \in [d], \forall x \in [0, 1]^d : \quad \partial_{ij} f(x) \le 0 
    ~.
\]

%\textbf{(Zhiyi: There shall be a direct reference for the next lemmas.)}
For any $y, z \in [0, 1]^d$, we write $y \vee z = \big( \max \{y_k, z_k\} \big)_{k \in [d]}$ and $y \wedge z= \big( \min \{y_k, z_k\} \big)_{k \in [n]}$ for the entry-wise maximum and minimum of two vectors.

\begin{lemma}[Submodularity, c.f., \citet{BianMBK:AIS:2017}]
    \label{lem:submodular}
    Suppose that $f : [0, 1]^d \to \R$ is a DR submodular function.
    Then, for any $y, z \in [0, 1]^d$:
    \[
        f(y) + f(z) \ge f(y \vee z) + f(y \wedge z)
        ~.
    \]
\end{lemma}

For any $y, z \in [0, 1]^d$ write $y \ge z$ if $y_k \ge z_k$ for any $k \in [d]$.
The next lemma may be folklore, but we cannot find a direct reference.
We therefore include its proof in Appendix~\ref{app:jensen-submodular} for completeness.

\begin{lemma}[Jensen Inequality for DR Submodular Functions on Ordered Points]
    \label{lem:jensen-submodular}
    Suppose that $y : [0, 1] \to [0, 1]^d$ satisfies that for any $0 \le \mu < \mu' \le 1$, $y(\mu) \le y(\mu')$, and $f : [0, 1]^d \to \R$ is a DR submodular function.
    Then:
    \[
        \int_0^1 f\big(y(\mu)\big) d\mu \le f \Big( \int_0^1 y(\mu) d\mu \Big)
        ~.
    \]
\end{lemma}

A function $f : [0, 1]^d \to \R$ is \emph{normalized} if $f(0, 0, \dots, 0) = 0$.

\begin{theorem}[Second Level Converse Jensen Inequality]
    \label{thm:2nd-converse-jensen}
    Suppose that $(x_{ij})_{(i,j) \in E}$ is in the polytope of the second level Poisson Matching LP, and $f : [0, 1]^2 \to \R$ is a normalized and DR submodular function.
    For any $j_1 \ne j_2 \in J$, let $\lambda_1^*, \lambda_2^*$ be the unique solution to:
    \begin{align*}
        x_{j_1} & = 1 - P_0(\lambda_1^*) ~; \\[1ex]
        x_{j_2} & = \big( 2 - P_0(\lambda_2^*) - P_1(\lambda_2^*) \big) - \big( 1 - P_0(\min\{\lambda_1^*, \lambda_2^*\}) \big) ~.
    \end{align*}
    Recall that $\rho_{ij} = \frac{x_{ij}}{\lambda_i}$ for any $(i,j) \in E$.
    We have:
    \[
        \sum_i \lambda_i f(\rho_{ij_1}, \rho_{ij_1}+\rho_{ij_2}) \ge 
        \begin{cases}
            \displaystyle
            \int_0^{\lambda_1^*} f \big( P_0(\lambda), P_1(\lambda) \big) d\lambda + \int_{\lambda_1^*}^{\lambda_2^*} f \big( 0, P_1(\lambda) \big) d\lambda
            & \text{if } \lambda_1^* \le \lambda_2^* ~; \\[3ex]
            \displaystyle
            \int_0^{\lambda_2^*} f \big( P_0(\lambda), P_1(\lambda) \big) d\lambda + \int_{\lambda_2^*}^{\lambda_1^*} f \big( P_0(\lambda), P_0(\lambda) \big) d\lambda
            & \text{if } \lambda_1^* > \lambda_2^* ~.
        \end{cases}
    \]
\end{theorem}

The theorem follows by making a sequence of adjustments to $(\rho_{ij_1})_{i \in I}$ and $(\rho_{ij_2})_{i \in I}$  subject to the Poisson constraints for the singleton set $\{j_1\}$ and for set $\{j_1, j_2\}$, so that $(\rho_{ij_1})_{i \in I}$ and $(\rho_{ij_2})_{i \in I}$ are closer to their values on the right-hand-side after each adjustment, and the weighted sum of function values on the left weakly decreases.
We defer the proof to Appendix~\ref{app:2nd-converse-jensen}.

\subsection{Zeroth Level Analysis: Matching Polytope}
\label{sec:0th-level}

We start from the simplest instantiation of the meta analysis, which only assumes that the fractional matching $(x_{ij})_{(i,j) \in E}$ is from the \emph{matching polytope}:
\begin{equation}
    \label{eqn:matching-polytope}
    \forall i \in I: \sum_{j \in J} x_{ij} \le \lambda_i
    ~, \quad
    \forall j \in J: \sum_{i \in I} x_{ij} \leq 1
    ~, \quad
    \forall (i, j) \in E: x_{ij} \ge 0
    ~.
\end{equation}

\begin{lemma}
    \label{lem:0th-level}
    Suppose that $(x_{ij})_{(i,j) \in E}$ is in the matching polytope.
    Then, the Poisson OCS ensures that for any $T \subseteq J$ and any $0 \le t \le 1$:
    \[
        \bar{Y}_T(t) \le \exp \Big( - t \sum_{j \in T} x_j \Big)
        ~.
    \]
\end{lemma}

\begin{proof}
    It suffices to show that functions $p_T(t) = \exp \big( - t \sum_{j \in T} x_j \big)$ for $T \subseteq J$ satisfy Equation~\eqref{eqn:meta-analysis}.
    With these functions, and by $\sum_j \rho_{ij} \le 1$, the right-hand side (omitting the negative sign) is:
    \[
        \exp \Big( - t \sum_{j \in T} x_j \Big) \sum_i \lambda_i \frac{\sum_{j \in T} e^{tx_j} \rho_{ij}}{\sum_{j \in T} e^{tx_j} \rho_{ij} + \sum_{j \notin T} \rho_{ij}} \ge \exp \Big( - t \sum_{j \in T} x_j \Big) \sum_i \lambda_i \frac{\sum_{j \in T} e^{tx_j} \rho_{ij}}{\sum_{j \in T} (e^{tx_j} - 1) \rho_{ij} + 1}
        ~.
    \]

    Further:
    \[
        \frac{\sum_{j \in T} e^{tx_j} \rho_{ij}}{\sum_{j \in T} (e^{tx_j} - 1) \rho_{ij} + 1} = \frac{(\sum_{j \in T} (e^{tx_j} - 1) \rho_{ij}) (1 - \sum_{j \in T} \rho_{ij} )}{\sum_{j \in T} (e^{tx_j} - 1) \rho_{ij} + 1} + \sum_{j \in T} \rho_{ij}
        \ge \sum_{j \in T} \rho_{ij}
        ~.
    \]

    Hence, the right hand side is at most:
    \[
        - \exp \Big( - t \sum_{j \in T} x_j \Big) \sum_i \lambda_i \sum_{j \in T} \rho_{ij}
        = - \exp \Big( - t \sum_{j \in T} x_j \Big) \sum_{j \in T} x_j
        ~.
    \]
    which equals the left-hand-side
\end{proof}

This implies that for any $j \in J$, $1-\bar{Y}_j(1) \geq 1-e^{-x_j} \geq (1-\frac{1}{e})x_j$.
Hence, the Poisson OCS with Matching LP is $(1-\frac{1}{e})$-competitive for unweighted and vertex-weighted online stochastic matching.

\subsection{First Level Analysis: Natural Polytope}
\label{sec:1st-level}

Next, we further assume that $(x_{ij})_{(i,j) \in E}$ is in the polytope of the Natural LP.
This polytope is a subset of the matching polytope so Lemma~\ref{lem:0th-level} still holds.
%Using the Poisson constraint~\eqref{eqn:natural-lp} and the 
Using the Converse Jensen Inequality (Theorem~\ref{thm:converse-jensen}), we prove a stronger bound for singleton sets.

\begin{theorem}
    \label{thm:1st-level}
    Suppose that $(x_{ij})_{(i,j) \in E}$ is in the polytope of the Natural LP.
    Then, Poisson OCS matches each offline vertex $j$ with probability at least $0.707\,x_j$ by the end, and therefore is $0.707$ competitive for unweighed and vertex-weighted online stochastic matching.
    %is at least $0.707$ competitive for unweighensures that for any $j \in J$:
    %
    %\[
    %    1 - \bar{Y}_j(1) \ge 0.7075\,x_j 
    %    ~.
    %\]
    %
\end{theorem}

\begin{proof}
    We will use the zeroth level bounds for all subsets $T$ with at least two offline vertices, and define $p_j(t)$ for individual offline vertices $j$ by recursively:
    %In order to prove this theorem, we first construct $p_T(t)$ satisfying the conditions in Lemma \ref{lem:meta-analysis} for all subsets $T\subseteq J$ in natural polytope. Let $p_T(t)$ be
    %
    \begin{align}\label{eqn:1st-level-recurrence}
        p_T(t) ~ = ~ &\exp \Big( - t \sum_{j \in T} x_j \Big) ~, & \qquad \forall T\subseteq J, |T|\geq 2
        ~; \notag\\
        p_j(0) ~ = ~ & 1 ~, & \forall j\in J
        ~; \notag \\
        \frac{d}{dt} p_j(t) ~ = ~ & p_j(t) \, \frac{\log \big( 1 - x_j + x_j \cdot \frac{e^{-2tx_j}}{ p_j(t)} \big)} { 1 - \frac{e^{-2tx_j}}{p_j(t)}} ~,\footnotemark[4]%
        & \forall j\in J~.
    \end{align}

    \footnotetext[4]{Define it to be $-p_j(t) x_j$ at the boundary case when $p_j(t) = e^{-2x_j t}$.}

    We next verify for any offline vertex $j$ that function $p_j(t)$ satisfies Eqn.~\eqref{eqn:meta-analysis}.
%\begin{lemma}
%    \label{lem:1st-p-function}
%    The $p_j(t)$ defined above satisfies Equation \eqref{eqn:meta-analysis}.
%\end{lemma}
%\begin{proof}
    First we state a lower bound of $p_j(t)$ to ensure non-negativity of the denominator.
    Appendix~\ref{app:1st-level-loose-bound} includes the proof.

    \begin{lemma}
        \label{lem:1st-level-loose-bound}
        For any $j \in J$ and any $0 \le t \le 1$, $p_j(t) \ge e^{-2tx_j}$.
    \end{lemma}

    Next by $p_{\{j, j'\}} = e^{-t(x_j+x_{j'})}$, the right-hand-side of Eqn.~\eqref{eqn:meta-analysis} is (omitting the negative sign):
    \[
        p_j(t) \sum_i \lambda_i \frac{\rho_{ij}}{\rho_{ij} + \frac{e^{-2tx_j}}{p_j(t)} \sum_{j' \neq j} \rho_{i{j'}}}\\
        ~ \ge ~
        p_j(t) \sum_i \lambda_i \frac{\rho_{ij}}{\rho_{ij} + \frac{e^{-2tx_j}}{p_j(t)}(1 - \rho_{ij})}
        ~.
    \]
    
    By the Converse Jensen Inequality (Theorem~\ref{thm:converse-jensen}) with $f(x)=-\frac{x}{x+Q\cdot(1-x)}$ where $Q=\frac{e^{-2t x_j}}{p_j(t)} \le 1$ (Lemma~\ref{lem:1st-level-loose-bound}), the right-hand-side of Eqn.~\eqref{eqn:meta-analysis} is at most:
    \[
        p_j(t) \int_{0}^{-\ln(1-x_j)}-\frac{e^{-\lambda}}{e^{-\lambda}+Q\cdot(1-e^{-\lambda})} d\lambda
        ~ = ~
        p_j(t)\frac{\log\left(1-x_j+Q\cdot x_j\right)}{1-Q}
        ~,
    \]
    which equals the left-hand-side of Eqn.~\eqref{eqn:meta-analysis} according to our recurrence \eqref{eqn:1st-level-recurrence}.
    %Therefore Equation \eqref{eqn:meta-analysis} holds.
%\end{proof}

    Since the recurrence~\eqref{eqn:1st-level-recurrence} is determined by $x_j$ but no other information about $j$, we instead consider for any $x \in [0, 1]$:
    \begin{equation}
        \label{eqn:1st-level-recurrence-x}
        p_x(0) = 1
        ~;\quad
        \forall\,0 \le t \le 1, ~
        \frac{d}{dt} p_x(t) = p_x(t) \, \frac{\log \big( 1 - x + x \cdot \frac{e^{-2xt}}{ p_x(t)} \big)} { 1 - \frac{e^{-2xt}}{p_x(t)}}
        ~.
    \end{equation}

    It remains to show that $\min_{x\in[0,1]} \frac{1-p_x(1)}{x} \ge 0.707$.
Calculating it numerically at $x = 1$ gives $1 - p_1(1) \approx 0.7075 > 0.707$ as desired.
    We next prove that the ratio is defined by $x = 1$.
    Consider:
%By Lemma \ref{lem:meta-analysis}, the competitive ratio of the algorithm is lower bounded by $\min_{x_j\in[0,1]} \frac{1-p_j(1)}{x_j}$. Here we show that the ratio reaches minimum when $x=1$. From the perspective of ratio, Equation \eqref{eqn:1-th level} is equivalent to:
    \[
        \frac{d}{dt} \frac{1-p_x(t)}{x} = - \, p_x(t) \, \frac{\log \big( 1 - x + x \cdot \frac{e^{-2tx}}{ p_x(t)} \big)} { x \big( 1 - \frac{e^{-2tx}}{p_x(t)} \big)}
        ~.
    \]

    For any time $0 \le t \le 1$ and any $0 \le x \le 1$, consider $y = \frac{1-p_x(t)}{x}$.
    Applying Lemma~\ref{lem:1st-level-loose-bound} and $\log(1-z) \le -z$ to recurrence~\eqref{eqn:1st-level-recurrence-x}, we have $\frac{d}{dt} p_x(t) \leq -p_x(t) x$ and  therefore $p_x(t) \le e^{-xt}$.
    This further means that $y\geq \frac{1-e^{-xt}}{x} \ge 1-e^{-t}$.
    We have the next lemma, whose proof is in Appendix~\ref{app:x-maximum}.

    \begin{lemma}
        \label{lem:x-maximum}
        For any $t\in[0,1]$, any $x \in [0, 1]$, and any $y \geq 1-e^{-t}$:
        \[
            (1-yx) \, \frac{\log\big(1-x+x\cdot\frac{e^{-2xt}}{ 1-yx}\big)}{x\big(1-\frac{e^{-2x t}}{1-yx}\big)}
            \le 
            (1-y) \, \frac{\log\big(\frac{e^{-2t}}{1-y}\big)}{1-\frac{e^{-2t}}{1-y}}
            ~,
        \]
        with equality at $x = 1$.
    \end{lemma}

    Hence, conditioned on any current value of $\frac{1-p_x(t)}{x}$, it decreases the fastest when $x=1$.
\end{proof}

\subsection{Second Level Analysis: Second Level Poisson Matching Polytope}
\label{sec:2nd-level}

Our final result for unweighted and vertex-weighted matching assumes that $(x_{ij})_{(i,j) \in E}$ lies in the polytope of the second level Poisson Matching LP.
Since it is a subset of the zeroth and first level polytopes, results from the last two subsections still hold.
%This polytope lies within the Natural LP so Lemma~\ref{thm:1st-level} still holds.
Using the first and second level Converse Jensen Inequality (Theorems~\ref{thm:converse-jensen} and \ref{thm:2nd-converse-jensen}), we next prove the main result of the section.

\begin{theorem}
    \label{thm:2nd-level}
    Suppose that $(x_{ij})_{(i,j) \in E}$ is in the polytope of the second level Poisson Matching LP.
    Then, Poisson OCS matches each offline vertex $j$ with probability at least $0.716\,x_j$ by the end, and therefore is $0.716$ competitive for unweighed and vertex-weighted online stochastic matching.
\end{theorem}

\begin{proof}
    We will use the zeroth level bounds for all subsets $T$ with at least three offline vertices:
    \[
        p_T(t) = \exp \Big( - t \sum_{j \in T} x_j \Big) ~, \qquad \forall\,T\subseteq J \mbox{ s.t. } |T|\geq 3
        ~.
    \]

    We next define a family of functions $p_T(t)$ for $|T| = 2$ that satisfy Eqn.~\eqref{eqn:meta-analysis}.
    The proofs of all lemmas are deferred to the end of the subsection.

    \begin{lemma}
        \label{lem:2nd-level-doubleton}
        Suppose that for any $1 \ge x_1 \ge x_2 \ge 0$ and the corresponding $\lambda_1^*, \lambda_2^*$ satisfying:
        \begin{align*}
            x_1 & = 1 - P_0(\lambda_1^*) ~, \\[1ex]
            x_2 & = \big( 2 - P_0(\lambda_2^*) - P_1(\lambda_2^*) \big) - \big( 1 - P_0(\min\{\lambda_1^*, \lambda_2^*\}) \big) ~,
        \end{align*}
        function $d_{x_1, x_2} : [0, 1] \to [0, 1]$ satisfies that:
        \begin{align*}
            d_{x_1, x_2}(0) & = 1 ~; \\
            \frac{d}{dt} d_{x_1, x_2}(t) &
            = - \, d_{x_1, x_2}(t) \cdot 
            \begin{cases}
                \displaystyle
                \int_0^{\lambda_1^*} f_{x_1,x_2} \big( P_0(\lambda), P_1(\lambda) \big) d\lambda + \int_{\lambda_1^*}^{\lambda_2^*} f_{x_1,x_2} \big( 0, P_1(\lambda) \big) d\lambda
                & \mbox{if } \lambda_1^* \le \lambda_2^* ~, \\[3ex]
                \displaystyle
                \int_0^{\lambda_2^*} f_{x_1,x_2} \big( P_0(\lambda), P_1(\lambda) \big) d\lambda + \int_{\lambda_2^*}^{\lambda_1^*} f_{x_1,x_2} \big( P_0(\lambda), P_0(\lambda) \big) d\lambda
                & \mbox{if } \lambda_1^* > \lambda_2^* ~,
            \end{cases}
        \end{align*}
        where (recall that $z^+ = \max \{z, 0\}$):
        \[
            f_{x_1,x_2}(y, z) = \frac{e^{t(2x_1+x_2)} d_{x_1, x_2}(t) \, y + e^{t(x_1+2x_2)} d_{x_1, x_2}(t) \, (z-y)} { \big(e^{t(2x_1+x_2)} d_{x_1, x_2}(t) - 1\big)^+ \, y + \big(e^{t(x_1+2x_2)} d_{x_1, x_2}(t) - 1 \big)^+\, (z-y) + 1}
            ~.
        \]
        Then, for any $T = \{j_1, j_2\} \subseteq J$ where $x_{j_1} \ge x_{j_2}$, $p_T(t) = d_{x_{j_1}, x_{j_2}}(t)$ satisfies Equation~\eqref{eqn:meta-analysis}.
    \end{lemma}

    In other words, we use the second level Converse Jensen Inequality to get an improved bound when $T$ is a doubleton.
    To propagate the improvement to individual offline vertices, we consider functions $q_x : [0, 1] \to [0, 1]$ for $x \in [0, 1]$ such that for any $x' \in [0, 1]$:
    \begin{equation}
        \label{eqn:2nd-level-q-condition}
        q_x(t) \ge 
        \begin{cases}
            e^{-tx'} d_{x, x'}(t) & \mbox{if $x \ge x'$;} \\[1ex]
            e^{-tx'} d_{x', x}(t) & \mbox{if $x < x'$.}
        \end{cases}
    \end{equation}

    This allows us have the same coefficients for all $j \notin T$ in the denominator in Eqn.~\eqref{eqn:meta-analysis} for singleton $T$, as we will see in the proof of the next lemma.

    \begin{lemma}
        \label{lem:2nd-level-singleton}
        Suppose that for any $0 \le x \le 1$, function $s_x : [0, 1] \to [0, 1]$ satisfies:
        \begin{align*}
            s_x(0) & = 1 ~; \\
            \frac{d}{dt} s_x(t) &
            = \frac{s_x(t) \log \left( 1 - x + x \cdot \frac{e^{-t x} q_x(t)}{s_x(t)} \right)} { 1 - \frac{e^{-t x} q_x(t)}{s_x(t)}}
            ~.
        \end{align*}
        Then, for any $j \in J$, $p_j(t) = s_{x_j}(t)$ satisfies Equation~\eqref{eqn:meta-analysis}.
    \end{lemma}

    It remains to find functions $\big(d_{x_1, x_2}\big)_{1 \ge x_1 \ge x_2 \ge 0}$, $\big(q_x\big)_{0 \le x \le 1}$, and $\big(s_x\big)_{0 \le x \le 1}$ that on the one hand satisfy the conditions in Lemmas~\ref{lem:2nd-level-doubleton} and \ref{lem:2nd-level-singleton} and in Equation~\eqref{eqn:2nd-level-q-condition}, and on the other hand ensure that for any $0 \le x \le 1$:
    \begin{equation}
        \label{eqn:2nd-level-ratio}
        1 - s_x(1) \ge 0.716 x
        ~.
    \end{equation}

    While we do not have the functions in closed forms, Appendix~\ref{app:2nd-numerical} explains how to numerically verify it for any given $0 \le x \le 1$ by an appropriate discretization of the functions.
    Further, once we have verified it for a finite yet sufficiently dense subset of $[0, 1]$ with ratios strictly better than $0.716$, we also cover the other values of $x$ in between.
\end{proof}

We next prove that the functions $f_{x_1, x_2}$ defined above satisfy the condition of the second level Converse Jensen Inequality.
This allows us to prove Lemma~\ref{lem:2nd-level-doubleton}.

\begin{lemma}
    \label{lem:f-dr-submodular}
    For any $1 \ge x_1 \ge x_2 \ge 0$, function $f_{x_1,x_2}$ is normalized and DR submodular.
\end{lemma}

\begin{proof}
    It follows by the definition that $f_{x_1, x_2}(0, 0) = 0$.
    Next we prove DR submodularity.
    For ease of presentation, let $A = e^{t(2x_1+x_2)} d_{x_1, x_2}(t)$ and $B = e^{t(x_1+2x_2)} d_{x_1, x_2}(t)$;
    we shall simply write $f$ for $f_{x_1, x_2}$.
    By $x_1 \ge x_2$ we have that $A \ge B$.
    The function then simplifies as:
    \[
        f(y, z) = \frac{ A \, y + B \, (z - y) }{(A-1)^+ \, y + (B-1)^+ \, (z-y) + 1}
        ~.
    \]

    If $A \ge B \ge 1$:
    \begin{align*}
        \partial_{y^2} f(y, z) & ~ = ~ - \, \frac{2 (A - B)^2 (1 - z)}{\big((A-B) \, y + (B-1) \, z + 1\big)^3} ~ \le 0
        ~, \\
        \partial_{y z} f(y, z) & ~ = ~ - \, \frac{2 (B-1) (B + (A-B) y)}{\big((A-B) \, y + (B-1) \, z + 1\big)^3} ~ \le 0
        ~, \\
        \partial_{z^2} f(y, z) & ~ = ~ - \, \frac{(A-B)((A-B)y+Bz+3(1-z))}{\big((A-B) \, y + (B-1) \, z + 1\big)^3} ~ \le 0
        ~.
    \end{align*}
    
    If $A \ge 1 > B$, i.e., if $z$'s coefficient in the denominator is $0$, we have:
    \begin{align*}
        \partial_{y^2} f(y, z) & ~ = ~ - \, \frac{2 \big( (1-B)+(A-1)(1-Bz) \big)}{\big((A-1)y + 1\big)^3} ~ \leq 0
        ~, \\
        \partial_{y z} f(y, z) & ~ = ~ - \, \frac{(A-1)B}{\big((A-1)y + 1\big)^2} ~ \leq 0
        ~, \\[1ex]
        \partial_{z^2} f(y, z) & ~ = ~ 0.
    \end{align*}
    
    Finally if $A, B < 1$, $f$ is linear and therefore is DR submodular.
\end{proof}

\begin{proof}[Proof of Lemma~\ref{lem:2nd-level-doubleton}]
    Up to renaming, we will assume without loss of generality that $j_1 = 1$ and $j_2 = 2$, for ease of notations.
    %We will write $x_1, x_2$ for $x_{j_1}, x_{j_2}$, and $\rho_{i1}, \rho_{i2}$ for $\rho_{ij_1}, \rho_{ij_2}$ for ease of notations.
    Applying the zeroth level bound $p_{1,2,j}(t) = e^{-t(x_1+x_2+x_j)}$ for any $j \ne 1, 2$ to the denominator of Eqn.~\eqref{eqn:meta-analysis} with $T = \{ 1, 2 \}$, and multiplying both the numerator and denominator by $e^{t(x_1+x_2)}$, it reduces to (recall that $\sum_j \rho_{ij} \le 1$):
    \[
        \frac{d}{dt} p_T(t) \ge - \, p_T(t) \sum_i \lambda_i \, \frac{e^{t(2x_1 + x_2)} p_T(t) \rho_{i1} + e^{t(2x_1 + x_2)} p_T(t) \rho_{i2}}{e^{t(2x_1 + x_2)} p_T(t) \rho_{i1} + e^{t(2x_1 + x_2)} p_T(t) \rho_{i2} + (1-\rho_{i1}-\rho_{i2})}
        ~.
    \]

    Since the coefficients of $\rho_{i1}, \rho_{i2}$ in the denominator are $e^{t(2x_1 + x_2)} - 1 \le (e^{t(2x_1 + x_2)} - 1)^+$, and $e^{t(x_1 + 2x_2)} - 1 \le (e^{t(x_1 + 2x_2)} - 1)^+$, it suffices to prove:
    \[
        \frac{d}{dt} p_T(t) \ge - \, p_T(t) \sum_i \lambda_i f_{x_1, x_2} \big( \rho_{i1}, \rho_{i1}+\rho_{i2} \big)
        ~.
    \]
    which follows by the second level Converse Jensen Inequality and the recurrence of $p_T = d_{x_1, x_2}$.
\end{proof}

Next we show that our new bound for any doubleton $T$ is indeed an improvement compared to the first level analysis.

\begin{lemma}
    \label{lem:2nd-level-doubleton-trivial}
    For any $1 \ge x_1 \ge x_2 \ge 0$, and any $0 \le t \le 1$, $d_{x_1, x_2}(t) \le e^{-t(x_1+x_2)}$.
\end{lemma}

\begin{proof}
    Observe that $d_{x_1, x_2}(t) \ge e^{-t(x_1+x_2)}$ implies $\frac{d}{dt} d_{x_1, x_2}(t) \le - d_{x_1, x_2}(t) (x_1 + x_2)$ according to its recurrence, and that it holds with equality at $t = 0$.
    The rest is a standard argument in analysis which we omit. 
    See, e.g., the proof of Lemma~\ref{lem:meta-analysis} in Appendix~\ref{app:meta-analysis} for a similar argument.
\end{proof}

By the definition of $q_x(t)$ in Equation~\eqref{eqn:2nd-level-q-condition}, we have the following corollary.

\begin{corollary}
    \label{cor:2nd-level-q-trivial}
    For any $0 \le x \le 1$ and $0 \le t \le 1$, $q_x(t) \le e^{-tx}$.
\end{corollary}

\begin{lemma}
    \label{lem:2nd-level-Q}
    For any $0 \le x \le 1$ and $0 \le t \le 1$, $s_x(t) \ge e^{-tx} q_x(t)$.
\end{lemma}

This is the counterpart of Lemma~\ref{lem:1st-level-loose-bound} from the first level analysis, except this time we have $q_x(t) \le e^{-tx}$ (Corollary~\ref{cor:2nd-level-q-trivial}) instead of exactly equal.
We omit the almost verbatim proof.
%The proof is almost verbatim which we omit.

%Then, we will complete the seciton by proving Lemmas~\ref{lem:2nd-level-doubleton} and \ref{lem:2nd-level-singleton}.

\begin{proof}[Proof of Lemma~\ref{lem:2nd-level-singleton}]
    The right-hand-side of the Eqn.~\eqref{eqn:meta-analysis} equals (omitting the negative sign):
    \begin{align*}
        p_j(t) \sum_i \lambda_i \frac{e^{tx_j} \rho_{ij} p_j(t)}{e^{tx_j} \rho_{ij} p_j(t) + \sum_{j' \neq j} e^{tx_{j'}} \rho_{ij'} p_{\{j,j'\}}(t)} 
        \ge ~ & p_j(t) \sum_i \lambda_i \frac{e^{tx_j} \rho_{ij} p_j(t)}{e^{tx_j} \rho_{ij} p_j(t) + \sum_{j' \neq j} q_{x_j}(t) \rho_{ij'}} \\
        \ge ~ & p_j(t) \sum_i \lambda_i \frac{e^{tx_j} \rho_{ij} p_j(t)}{e^{tx_j} \rho_{ij} p_j(t) + q_{x_j}(t) (1-\rho_{ij})}
        ~.
    \end{align*}
        
    The lemma then follows by the first level Converse Jensen Inequality (with convexity ensured by Lemma~\ref{lem:2nd-level-Q}), and the recurrence for $p_j(t) = s_{x_j}(t)$.
\end{proof}

\clearpage

\section{Hardness Results}
\label{sec:hardness}

In this section, $\opt$ denotes the expected objective of the optimal matching for the realized bipartite graph, and $\algo$ denotes the expected objective of the online algorithm.

\subsection{Hardness for Edge-weighted Matching without Free Disposal}
\label{sec:hardness-edge-weighted}

In unweighted matching, vertex-weighted matching, and edge-weighted matching \emph{with free disposal}, it is without loss of generality to match every online vertex whenever possible.
In edge-weighted matching \emph{without free-disposal}, however, we may want to leave an online vertex unmatched even if it has an unmatched neighbor, so that the neighbor will be available for a possible heavier edge later.
This is another disadvantage of the online algorithm, because the offline optimal knows the realization of online vertices and correctly decides whether to match the lighter edge.
%, there is a gap between them even for very simple cases, such as two types of online vertices connecting to one offline vertex with different edge weights.
We combine this with the hard instance of \citet{ManshadiOS:MOR:2012} to obtain a harder instance for edge-weighted matching without free disposal, separating it from the other three settings.
%hardness  Combining the above idea with this instance, for each offline vertex, we insert a new type of online vertices only adjcent to it, with a small arrival rate and a large weight, which can contribute all its value to the offline solution but not the online algorith

\begin{theorem}
    There is an instance of edge-weighted online stochastic matching model without free disposal for which no algorithm has a competitive ratio better than $0.703$.
\end{theorem}

\begin{proof}
    The instance is online-vertex-weighted, which means that each online vertex type has a positive weight $w_i$ and $w_{ij}=w_i$ for all its adjacent edges $(i,j)$.

    Consider $|J| = n$ offline vertices for a sufficiently large $n$.
    There are four kinds of online types $I = I_1 \sqcup I_2 \sqcup I_3 \sqcup I_n$.
    For $k \in \{1, 2, 3, n\}$, $I_k$ contains $\binom{n}{k}$ online types, each adjacent to a different subset of $k$ offline vertices.
    Following the instance of \citet{ManshadiOS:MOR:2012}, consider $m = \frac{1}{2} c^*_{2.5} n$ where $c^*_{2.5} \approx 0.81$ is a constant from \citet{DietzfelbingerGMMPR:ICALP:2010}, and let online types in $I_2$ and $I_3$ have unit weights and arrival rates $m/\binom{n}{2}$ and $m/\binom{n}{3}$ respectively.
    Further, let $\varepsilon$ be an infinitesimal constant, and let $x$ be a constant to be determined.
    The unique online type in $I_n$ also has unit weight, and arrives at rate $n-2m-n\varepsilon$.%
    \footnote{We let it be $n-2m-n\varepsilon$ instead of $n-2m$ so that the arrive rates sum to $n$, making the hard instance valid in the original online stochastic matching model. It could be $n-2m$ if we only consider the Poisson arrival model.}
    Finally, let the online types in $I_1$ have weights $\frac{x}{\varepsilon}$ and arrival rates $\varepsilon$.
    See Figure~\ref{fig:hardness-edge-weighted-without-free-disposal} for an illustrative picture.
    
    On the one hand, offline optimal can match all online vertices with types in $I_2, I_3, I_n$ with high probability~\cite{ManshadiOS:MOR:2012, DietzfelbingerGMMPR:ICALP:2010}, when no online vertices have types in $I_1$.
    On the other hand, offline optimal can match each online vertex with type in $I_1$ (potentially making an online vertex with types in $I_2, I_3, I_n$ unmatched), and increases the objective by at least $\frac{x}{\varepsilon}-1 \approx \frac{x}{\varepsilon}$.
    Hence, $\opt = (1+x)n - o(n)$.

    %Let $m = \frac{1}{2}c^*_{2.5}n$ for some constant $c^*_{2.5}$. For $i\in I_1$, let $\lambda_i=\varepsilon$ and $w_i=\frac{x}{\varepsilon}$. For $i\in I_2$, let $\lambda_i=m/\binom{n}{2}$ and $w_i=1$. For $i\in I_3$, let $\lambda_i=m/\binom{n}{3}$ and $w_i=1$. For $i\in I_n$, let $\lambda_i=n-2m-n\varepsilon$ and $w_i=1$. The constant $c^*_{2.5}$ is defined to be that the offline solution can match all the online vertices in $I_2$ and $I_3$ with high probability (see \citet{ManshadiOS:MOR:2012} for more discussions of $c^*_k$). \citet{DietzfelbingerGMMPR:ICALP:2010} show that $c^*_{2.5}\approx 0.81$. 

    Next consider any online algorithm.
    It is easier to analyze it under the original online stochastic model;
    the result applies to the Poisson arrival model as well by the asymptotic equivalence.
    Since the arrival rates sum to $\Lambda = n$, we next consider steps $t = 1, 2, \dots, n$ each of which has an online vertex drawn from the distribution.
    For $1 \le t \le n$, $A(t)$ denotes the number of matched offline vertices after step $t$;
    let $A(0) = 0$.
    Fixing any step $0 \le t < n$ and the value of $A(t)$, consider step $t+1$.
    For $k = 2, 3$, with probability $\frac{m}{n}$ the online vertex has type in $I_k$, and conditioned on that its neighbors are all matched with probability $\binom{A(t)}{k}/\binom{n}{k}$.
    Hence:
    \[
        A(t+1) \leq A(t) + 1 - \frac{m}{n} \frac{\binom{A(t)}{2}}{\binom{n}{2}} - \frac{m}{n} \frac{\binom{A(t)}{3}}{\binom{n}{3}}
        ~.
    \]

    Taking expectation on both sides and by the Jensen Inequality:
    \begin{equation}
        \label{eqn:hardness-edge-weighted-1}
        \E\,A(t+1)
        \leq \E\,A(t) + 1 - \frac{m}{n} \frac{\binom{\E A(t)}{2}}{\binom{n}{2}} - \frac{m}{n}\frac{\binom{\E A(t)}{3}}{\binom{n}{3}}
        ~.
    \end{equation}

    Let $B(t)$ denote the total weight of matched online vertices with types in $I_1$ after time $t$, with $B(0) = 0$.
    Fix any step $0 \le t < n$ and the value of $A(t)$, and consider step $t+1$.
    With probability $\varepsilon$ the online vertex has type in $I_1$.
    Conditioned on that, matching it adds $\frac{x}{\varepsilon}$ to $B(t+1)$, but with probability $1-\frac{A(t)}{n}$ the unique neighbor is already matched.
    Hence:
    \begin{equation}
        \label{eqn:hardness-edge-weighted-2}
        \E\,B(t+1) \le \E\,B(t) + x \Big(1-\frac{\E\,A(t)}{n} \Big)
        ~.
    \end{equation}

    By definition we have $\algo \leq \E\,A(n)+\E\,B(n)$.
    It remains to bound $\E\,A(n)+\E\,B(n)$ subject to Equations~\eqref{eqn:hardness-edge-weighted-1} and \eqref{eqn:hardness-edge-weighted-2}.
    First, it is without loss of generality to assume that Eqn.~\eqref{eqn:hardness-edge-weighted-2} holds with equality.
    We next prove that the maximum value can only be achieved when for all $0 \le t < n$, either $\E\,A(t) = 0$, or Eqn.~\eqref{eqn:hardness-edge-weighted-1} holds with equality.
    Suppose for contrary that for some $t$ we have $\E\,A(t) > 0$ but the left-hand-side of Eqn.~\eqref{eqn:hardness-edge-weighted-1} is strictly smaller than the right.
    %there is some $1 \le k \le n-1$ such that the equality does not hold in  and $\E A(t)>\E A(t-1)$, 
    We then decrease $\E\,A(t)$ by a sufficiently small amount so that Eqn.~\eqref{eqn:hardness-edge-weighted-1} still holds, with the same value for $\E\,A(n)$.
    By Eqn.~\eqref{eqn:hardness-edge-weighted-2} with equality, on the other hand, the value of $\E\,B(n)$ strictly increases.

    Finally, to numerically bound the maximum of $\E\,A(n)+\E\,B(n)$, we can enumerate $1 \le k \le n$, and for each $k$ consider $\E\,A(t) = 0$ for $t < k$ and for some value of $0 \le \E\,A(k) \le 1$ recursively compute $\E\,A(t)$ for $k < t \le n$ by Eqn. \eqref{eqn:hardness-edge-weighted-1} with equality.
    In fact, we further assume $\E\,A(k)=0$ which introduces an absolute error of at most $1$ in the bound.
    %So we can suppose there exists $k\in[1,n]$ such that $\E A(t)=0$ for $t<k$ and the euqality holds in Eqn. \eqref{eqn:hardness-edge-weighted-1} for $t\geq k$. 
    %With an absolute error no more than $1$, we further suppose $\E A(k)=0$. 
    For $n=10^6$, $x=0.94$, the numerical bound shows that $\frac{\algo}{\opt}<0.703$, with the maximum value achieved when $k \approx 2.07 \times 10^5$.
\end{proof}

%This theorem together with Theorem \ref{thm:top-half-sample} show that the general edge-weighted online stochastic matching is strictly harder than that with free disposal.
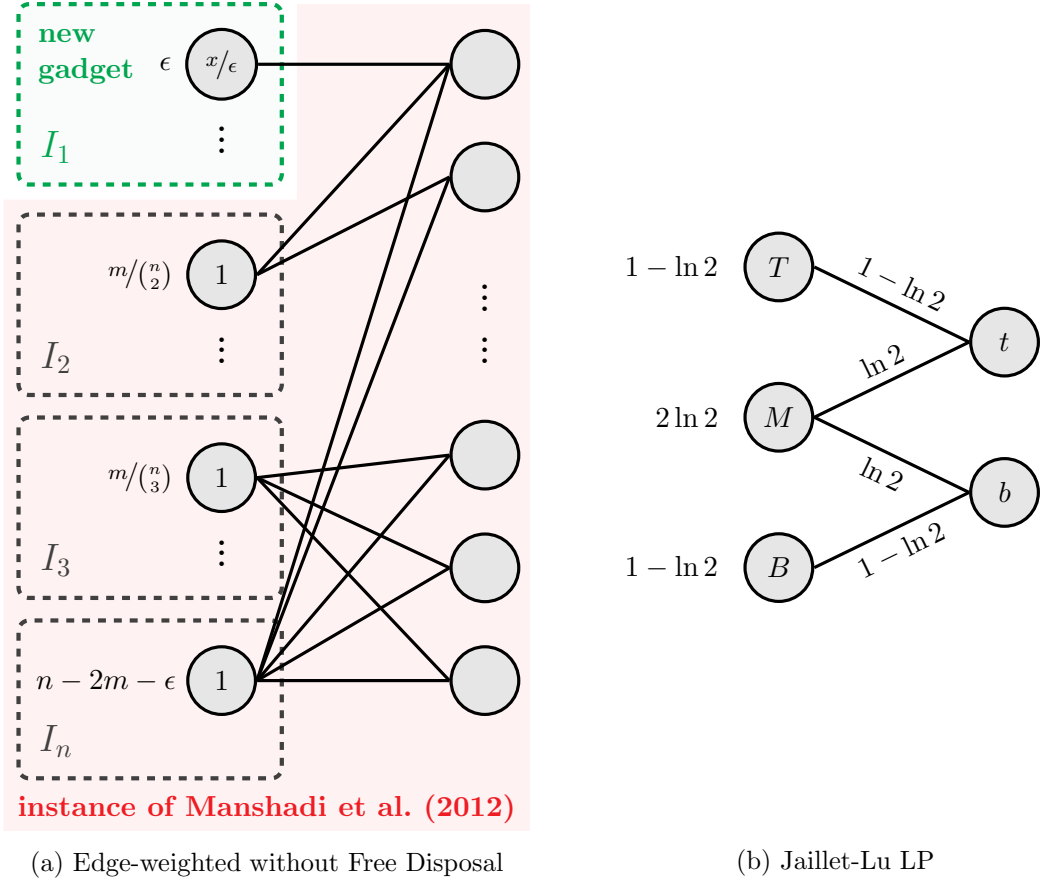
\begin{figure}[t]
    \tikzset{vertex/.style={draw=black,very thick,fill=gray!20,minimum size=.9cm,circle}}
    \tikzset{edge/.style={draw=black,very thick}}
    \centering
    \begin{subfigure}{.5\textwidth}
        \centering
        \begin{tikzpicture}
            \begin{scope}
                \draw[draw=Green,ultra thick,dashed,rounded corners,fill=Green!2] (.8,.8) rectangle +(-3.5,-2.4);
                \node[vertex] (1) at (0,0) {$\nicefrac{x}{\epsilon}$};
                \node[left=0.05cm of 1] {$\epsilon$};
                %\node[vertex] (N2) at (0,-1.3) {$\nicefrac{x}{\epsilon}$};
                \node at (0,-0.9) {\bf \vdots};
                \node at (-2.2,-1.1) {\textcolor{Green}{\Large $I_1$}};
                \node[align=left] at (-1.8,.1) {\textcolor{Green}{\bf new}\\\textcolor{Green}{\bf gadget}};
            \end{scope}
            \draw[draw=none,ultra thick,dashed,fill=red!5] (4.1,.8)--(4.1,-10.2)--(-2.9,-10.2)--(-2.9,-1.8)--(1,-1.8)--(1,.8)--(4.1,.8);
            \node at (0.6,-9.9) {\textcolor{Red}{\bf instance of Manshadi et al.~(2012)}};
            \begin{scope}[shift={(0,-2.8)}]
                \node[vertex] (2) at (0,0) {$1$};
                \node[left=0.05cm of 2] {$\nicefrac{m}{\binom{n}{2}}$};
                \node at (0,-0.9) {\bf \vdots};
                \draw[draw=darkgray,ultra thick,dashed,rounded corners] (.8,.8) rectangle +(-3.5,-2.4);
                \node at (-2.2,-1.1) {\textcolor{darkgray}{\Large $I_2$}};
            \end{scope}
            \begin{scope}[shift={(0,-5.5)}]
                \node[vertex] (3) at (0,0) {$1$};
                \node[left=0.05cm of 3] {$\nicefrac{m}{\binom{n}{3}}$};
                \node at (0,-0.9) {\bf \vdots};
                \draw[draw=darkgray,ultra thick,dashed,rounded corners] (.8,.8) rectangle +(-3.5,-2.4);
                \node at (-2.2,-1.1) {\textcolor{darkgray}{\Large $I_3$}};
            \end{scope}
            \begin{scope}[shift={(0,-8.2)}]
                \node[vertex] (n) at (0,0) {$1$};
                \node[left=0cm of n] {$n-2m-\epsilon$};
                \draw[draw=darkgray,ultra thick,dashed,rounded corners] (.8,.8) rectangle +(-3.5,-2.1);
                \node at (-2.2,-.8) {\textcolor{darkgray}{\Large $I_n$}};
            \end{scope}
            \begin{scope}[shift={(3.5,0)}]
                \node[vertex] (R1) at (0,0) {};
                \node[vertex] (R2) at (0,-1.5) {};
                \node at (0,-3) {\bf\vdots};
                \node at (0,-3.7) {\bf\vdots};
                \node[vertex] (R3) at (0,-5.2) {};
                \node[vertex] (R4) at (0,-6.7) {};
                \node[vertex] (R5) at (0,-8.2) {};
            \end{scope}
            \path[edge] (1.east) edge (R1.west);
            \path[edge] (2.east) edge (R1.west);
            \path[edge] (2.east) edge (R2.west);
            \path[edge] (3.east) edge (R3.west);
            \path[edge] (3.east) edge (R4.west);
            \path[edge] (3.east) edge (R5.west);
            \path[edge] (n.east) edge (R1.west);
            \path[edge] (n.east) edge (R2.west);
            \path[edge] (n.east) edge (R3.west);
            \path[edge] (n.east) edge (R4.west);
            \path[edge] (n.east) edge (R5.west);
        \end{tikzpicture}
        \caption{Edge-weighted without Free Disposal}
        \label{fig:hardness-edge-weighted-without-free-disposal}
    \end{subfigure}
    \begin{subfigure}{.4\textwidth}
        \centering
        \begin{tikzpicture}
            \draw[fill=none,draw=none] (-2.5,3.5) rectangle +(6.5,-11);
            \node[vertex] (T) at (0,0) {$T$};
            \node[left=.2cm of T] {$1-\ln 2$};
            \node[vertex] (M) at (0,-2) {$M$};
            \node[left=.2cm of M] {$2\ln 2$};
            \node[vertex] (B) at (0,-4) {$B$};
            \node[left=.2cm of B] {$1-\ln 2$};
            \node[vertex] (t) at (3,-1) {$t$};
            \node[vertex] (b) at (3,-3) {$b$};
            \path[edge] (T.east) edge node[midway,above,sloped] {$1-\ln 2$} (t.west);
            \path[edge] (M.east) edge node[midway,above,sloped] {$\ln 2$} (t.west);
            \path[edge] (M.east) edge node[midway,below,sloped] {$\ln 2$} (b.west);
            \path[edge] (B.east) edge node[midway,below,sloped] {$1-\ln 2$} (b.west);
        \end{tikzpicture}
        \caption{Jaillet-Lu LP}
        \label{fig:hardness-jl-lp}
    \end{subfigure}
    \caption{Illustration of Hard Instances. The number on the left of each online type is the arrival rate. On the left, the number inside a vertex is the vertex-weight. On the right, the letter inside a vertex is its name.}
\end{figure}

%\begin{figure}[t]
%    \centering
%    \begin{subfigure}{.4\textwidth}
%        \includegraphics[width=\textwidth]{figure-hardness-edge-weighted.pdf}
%        \caption{Edge-weighted without Free Disposal}
        %\label{fig:hardness-edge-weighted-without-free-disposal}
%    \end{subfigure}
    %
%    \begin{subfigure}{.4\textwidth}
%        \includegraphics[width=\textwidth]{figure-hardness-jl.pdf}
%        \caption{Jaillet-Lu LP}
        %\label{fig:hardness-jl-lp}
%    \end{subfigure}
%    \caption{Hard Instances}
%\end{figure}

\subsection{Hardness for the Jaillet-Lu Linear Program}

%The top half sampling algorithm does not fully utilize Equation \eqref{eqn:natural-lp} in the natural LP.
Since Top Half Sampling only needs the inequality in Corollary~\ref{cor:inverse-jensen}, its competitive ratio still holds if we use the LP of \citet{JailletL:MOR:2014} (instead of the Natural LP), which we restate below:
\begin{alignat}{2}
    \mbox{maximize}\quad & \sum_{(i,j)\in E} w_{ij} x_{ij} & {} & \nonumber \\
    \mbox{subject to}\quad & \sum_{j \in J_i} x_{ij}\leq\lambda_i &\quad& \forall i\in I \nonumber \\
    & \sum_{i \in I_j} x_{ij} \le 1 &\quad& \forall j\in J \nonumber \\
    & \sum_{i \in I_j} ( 2x_{ij} - \lambda_i )^+ \le 1 - \ln 2 &\quad& \forall j\in J \nonumber \\
    & x_{ij} \geq 0 &\quad& \forall(i,j)\in E \nonumber
\end{alignat}

In fact, our competitive ratio is the same as that of \citet{JailletL:MOR:2014} for unweighted matching, which has been the state-of-the-art until a very recent improvement by \citet{HuangS:STOC:2021}.
%Under the Jaillet-Lu LP, the competitive ratio of the amplified sampling algorithm proved in Theorem \ref{thm:top-half-sample} is actually optimal. 
We next show that this ratio is tight if we compare to the Jaillet-Lu LP.
In this sense our analysis is tight.
It also indicates that more expressive LPs such as the Natural LP and the Poisson Matching LPs are necessary for the better ratios in unweighted and vertex-weighted matching in Section~\ref{sec:vertex-weighted-ocs}.

\begin{theorem}
    There is an instance of unweighted online stochastic matching for which no algorithm gets $\frac{\algo}{\jl} > 1 - \frac{1}{1-\ln 2} \big( \frac{1}{2e} - \frac{\ln 2}{e^2} \big) \approx 0.706$, where $\jl$ is the optimal value of the Jaillet-Lu LP.
\end{theorem}

\begin{proof}
    We prove it with a small instance in the Poisson arrival model.
    By duplicating many copies of the instance and the asymptotic equivalence of models, the theorem holds in the original online stochastic matching model as well.

    Let $I=\{T,M,B\}$, $J=\{t,b\}$, where $T, t$ stand for top, $B, b$ stand for bottom, and $M$ stands for middle.
    Let $E=\{(T,t),(B,b),(M,t),(M,b)\}$.
    Online types $T, B$ have arrival rates $\lambda_T=\lambda_B=1-\ln2$.
    Online type $M$ has arrival rate $\lambda_M=2\ln2$.
    The Jaillet-Lu LP optimal is $x_{Tt}=x_{Bb}=1-\ln2$, $x_{Mt}=x_{Mb}=\ln2$, with objective  $\jl=2$.
    See Figure~\ref{fig:hardness-jl-lp} for an illustrative picture.

    We claim that the optimal algorithm simply matches each online vertex whenever possible and arbitrarily.
    For an online vertex of type $T$ or $B$, or an online vertex of type $M$ with only one unmatched neighbor when it arrives, it is trivially true.
    For an online vertex of type $M$ with both neighbors unmatched, this is still true by symmetry.
    Consider an optimal algorithm that break ties in favor of $t$.
    %Consider one of them that always prefers to match $j_1$ for $i_3$. 
    For any $i \in I$, let $n_i$ be the number of online vertices of type $i$.
    Further let $i_0$ be the type of the earliest online vertex.
    We have:
    \begin{align*}
        \algo
        &
        = 2 - \Pr \big[ \text{$t$ is unmatched} \big] - \Pr \big[ \text{$b$ is unmatched} \big] \\[2ex]
        &
        = 2 - \Pr \big[n_T=n_M=0\big] - \Big( \Pr \big[n_B=n_M=0\big] + \Pr \big[n_B=0, n_M=1, i_0=M\big] \Big)
        ~.
    \end{align*}

    By the Poisson arrival model:
    \[
        \Pr \big[n_T=n_M=0\big] = \Pr \big[n_B=n_M=0\big] = e^{-(1+\ln2)}
        ~.
    \]

    Finally:
    \begin{align*}
        \Pr \big[n_B=0, n_M=1, i_0=M\big]
        %&
        %= \Pr \big[n_B=0\big]\,\Pr \big[n_M=1\big]\,\Pr \big[i_0=M \mid n_B=0, n_M=1\big]
        %\\
        &
        = \underbrace{\vphantom{\bigg[}e^{-(1-\ln2)}}_{\Pr [n_B=0]} ~ \underbrace{\vphantom{\bigg[}(2\ln2) e^{-2\ln2}}_{\Pr[n_M=1]} \underbrace{\vphantom{\bigg[}\int_0^1 e^{-(1-\ln2)t} dt}_{\Pr [i_0=M \mid n_B=0, n_M=1]}
        \\[1ex]
        &
        = e^{-(1+\ln2)} \frac{2\ln2}{1-\ln2} \Big(1-\frac{2}{e}\Big)
        ~.
    \end{align*}

    Putting together gives:
    \[
        \algo = 2 \Big(1 - \frac{1}{1-\ln2}\Big(\frac{1}{2e}-\frac{\ln2}{e^2}\Big) \Big)
        ~.
    \]

    Comparing $\algo$ and $\jl$ proves the theorem.
\end{proof}

%For the online stochastic matching model, we can construct a graph with $k$ ($k\to\infty$) parts, where each part is isomorphic to the above instance. Since when $k\to\infty$, it's asymptotically the same with the Poisson arrival model, this hardness result also holds for the online stochastic matching model.

\section*{Acknowledgments}

We thank Donglei Du for helpful discussions on DR submodular functions.
We also thank Zipei Nie and Nengkun Yu for their help with the analysis of differential inequalities.

\bibliographystyle{plainnat}
\bibliography{reference}

\appendix

\section{Missing Proofs in Section~\ref{sec:edge-weighted-differential-inequality}}

\subsection{Proof of Lemma~\ref{lem:alg-init-rate}}
\label{app:alg-init-rate}

The definition of $\sigma_{i,t}$ ensures that $w_{i,\sigma_{i,t}(\theta)}(t)$ is non-increasing over $\theta$, so we have:
\begin{align*}
    \frac{d}{dt} \E \bar{\alg}(t=0)
    &
    = - \sum_{i\in I} \underbrace{\vphantom{\bigg[} \lambda_i}_\text{arrival rate of $i$} \cdot \underbrace{\vphantom{\bigg[} \frac{2}{\lambda_i} \int_{0}^{\lambda_i/2}w_{i,\sigma_{i,0}(\theta)}d\theta}_\text{expected gain from an arrival of $i$}
    \\
    &
    \le - \sum_{i\in I} \int_{0}^{\lambda_i}w_{i,\sigma_{i,0}(\theta)}d\theta
    \\
    &
    = - \sum_{i\in I} \sum_{j\in J_i} x_{ij}w_{ij}
    \\
    &
    = - \opt
    ~.
\end{align*}

\subsection{Missing Analysis of Differential Inequality from Section~\ref{sec:edge-weighted-differential-inequality}}
\label{app:edge-weighted}

Recall that the differential inequality is:
\[
    (2 + 2 \ln 2)\,\E\,\bar{\alg}(t) + (3 + \ln 2)\,\frac{d}{dt}\,\E\,\bar{\alg}(t) + \frac{d^2}{dt^2}\,\E\,\bar{\alg}(t) \le 0
    ~,
\]
and the boundary conditions are:
\[
    \bar{\alg}(0) = \opt
    ~,\quad
    \frac{d}{dt}\,\E\,\bar{\alg}(t=0) \le -\opt
    ~.
\]

Consider $B(t)$ that satisfies the differential inequlaity and boundary conditions with equalities:
\[
    B(t) = \frac{\opt}{1-\ln 2} \Big( \frac{1}{(2e)^t} - \frac{\ln 2}{e^{2t}} \Big)
    ~.
\]

We seek to prove that $\E \bar{\alg}(1) \le B(1)$.
In fact we will prove a stronger claim that $\E \bar{\alg}(t) \le B(t)$ for all $0 \le t \le 1$.
Let $C(t) = \E \bar{\alg}(t) - B(t)$.
The claim is then $C(t) \le 0$ for all $t$.

The differential inequality and equality for $\E \bar{\alg}(t)$ and $B(t)$ and their boundary conditions imply:
\[
    (2 + 2 \ln 2)\,C(t) + (3 + \ln 2)\frac{d}{dt}\,C(t) + \frac{d^2}{dt^2}C(t) \le 0
    ~,
\]
and boundary conditions:
\[
    C(0) = 0
    ~,\quad
    \frac{d}{dt}C(t=0) \le 0
    ~.
\]

Further consider $D(t) = 2C(t)+\frac{d}{dt}C(t)$.
The differential inequality and boundary conditions for $C(t)$ implies a differential inequality for $D(t)$:
\[
    (1+\ln 2) D(t) + \frac{d}{dt} D(t) \le 0
    ~,
\]
and its boundary condition:
\[
    D(0) \le 0
    ~.
\]

The differential inequality for $D$ gives:
\[
    \frac{d}{dt} (2e)^t D(t) = (2e)^t \Big( (1+\ln 2) D(t) + \frac{d}{dt} D(t) \Big) \le 0
    ~.
\]

Hence, for any $t$ we have $(2e)^t D(t) \le D(0) \le 0$.
Thus, $D(t) \le 0$, or equivalently:
\[
    2C(t)+\frac{d}{dt}C(t) \le 0
    ~.
\]

Similarly, this means that:
\[
    \frac{d}{dt} e^{2t} C(t) = e^{2t} \Big( 2 C(t) + \frac{d}{dt} C(t) \Big) \le 0
    ~,
\]
and therefore $e^{2t} C(t) \le C(0) \le 0$ for any $t$.
This gives $C(t) \le 0$ as desired.

\section{Missing Proofs in Section~\ref{sec:vertex-weighted-ocs}}
\label{app:vertex-weighted}

\subsection{Proof of Lemma~\ref{lem:meta-analysis}}
\label{app:meta-analysis}

We shall prove the lemma by an induction on the size of $T$ in descending order.
The base case is $T = J$.
Recall that $\Lambda = \sum_i \lambda_i$ is the total arrival rate.
The conditions about $\bar{Y}_J(t)$ simplify to:
\[
    \bar{Y}_J(0) = 1
    ~, \quad
    \frac{d}{dt} \bar{Y}_J(t) \le - \bar{Y}_J(t) \Lambda
    ~.
\]
which implies $\frac{d}{dt} e^{\Lambda t} \bar{Y}_J(t) \le 0$ and therefore $\bar{Y}_J(t) \le e^{-\Lambda t}$.
%The lemma follows by $\sum_j x_j \le \Lambda$.

Similarly, the condition about $p_J(t)$ simplify to:
\[
    p_J(0) = 1
    ~, \quad
    \frac{d}{dt} p_J(t) \ge - p_J(t) \Lambda
    \quad \text{almost everywhere,}
\]
which implies $\frac{d}{dt} e^{\Lambda t} p_J(t) \ge 0$ almost everywhere and therefore $p_J(t) \ge e^{-\Lambda t}$.

Next for some $n < |J|$ suppose that the inequality holds for any subset $T \subseteq J$ with $|T| = n+1$.
Consider any subset $T \subset J$ with $|T| = n$.
We first relax the differential inequality for $\bar{Y}_T(t)$ by bounding $\bar{Y}_{T+j}(t)$ using the inductive hypothesis:
\[
    \frac{d}{dt} \bar{Y}_T(t) \le - \bar{Y}_T(t) \sum_i \lambda_i \frac{\sum_{j \in T} e^{tx_j} \rho_{ij} \bar{Y}_T(t)}{\sum_{j \in T} e^{tx_j} \rho_{ij} \bar{Y}_T(t) + \sum_{j \notin T} e^{tx_j} \rho_{ij} p_{T+j}(t)}
    ~.
\]

The right-hand-side is decreasing in $\bar{Y}_T(t)$.
Hence, almost everywhere we get that $\bar{Y}_T(t) \ge p_T(t)$ implies $\frac{d}{dt} \bar{Y}_T(t) \le \frac{d}{dt} p_T(t)$.
Consider an auxiliary function $g(t) = p_T(t) - \bar{Y}_T(t)$ for $t \in [0, 1]$.
We have that $g(0) = 0$, and almost everywhere $g(t) \le 0$ implies $\frac{d}{dt} g(t) \ge 0$.
Further by absolutely continuity of $p_T$ and observing that $\bar{Y}_T(t)$ is also absolutely continuous (in fact, it is decreasing and $\Lambda$-Lipschitz because it cannot decrease faster than the total arrival rate of online vertices), $g$ is also absolutely continuous.
The inequality of the lemma is equivalent to $g(t) \ge 0$ for $0 \le t \le 1$.
Suppose for contrary that there is $0 \le t_0 \le 1$ such that $g(t_0) < 0$.
Consider $t_1 = \sup \{ 0 \le t \le t_0 : g(t) \ge 0 \}$.
We have $g(t_1) = 0$, $t_1 < t_0$, and $g(t) < 0$ for any $t_1 < t \le t_0$.
On the one hand, $g(t_0) - g(t_1) = g(t_0) < 0$.
On the other hand , $g(t_0) - g(t_1) = \int_{t_1}^{t_0} g'(t) dt \ge 0$ since $g(t) < 0$ implies $g'(t) \ge 0$ almost everywhere for $t_1 < t < t_0$.
We have a contradiction.

\subsection{Proof of Lemma~\ref{lem:poisson-matching-lp-relaxation}}
\label{app:poisson-matching-lp-relaxation}

For any $(i, j) \in E$, let $x_{ij}$ be the probability that offline vertex $j$ is matched to an online vertex of type $i$ in the optimal matching of the realized graph.
By definition, the expected objective of the optimal matching is:
\[
    \sum_{(i, j) \in E} w_{ij} x_{ij}
    ~.
\]

It remains to prove that this is feasible for the $\ell$-th level Poisson Matching LP for any $\ell \ge 0$.
Non-negativity holds trivially.
We next verify the other constraints.

\paragraph{Matching Constraints (Online).}
For any online type $i \in I$, $\sum_j x_{ij}$ is the expected number of matched online vertices of type $i$ in the optimal matching for the realized graph, which is upper bounded by the expected number of online vertices of type $i$, i.e.:
\[
    \sum_{j \in J} x_{ij} \le \lambda_i
    ~.
\]

\paragraph{Matching Constraints (Offline).}
This is relevant only at the zeroth level.
For any offline vertex $j \in J$, $\sum_i x_{ij}$ is the probability that $j$ is matched in the optimal matching for the realized graph, and therefore cannot exceed one, i.e.:
\[
    \sum_{i \in I} x_{ij} \le 1
    ~.
\]

\paragraph{Poisson Constraints ($\ell$-th Level).}
For any subset $S \subseteq I$ of online types, any subset $T \subseteq J$ of offline vertices such that $|T| = \ell$, $\sum_{i \in I} \sum_{j \in J} x_{ij}$ is the expected number of offline vertices in $T$ that are matched to some online vertices with types in $S$, in the optimal matching for the realized graph.
This is upper bounded by the expectation of the number of online vertices with types in $S$ \emph{capped by $|T| = \ell$}.
Recall that $1 - P_{k-1}(\Lambda_S)$ is the probability of having at least $k$ online vertices with types $S$ in the realized graph in the Poisson arrival model.
The aforementioned expectation can be written as $\sum_{k=1}^\ell \big( 1 - P_{k-1}(\lambda_S) \big)$.
Therefore:
\[
    \sum_{i \in S} \sum_{j \in T} x_{ij} \le \sum_{k=1}^\ell \big( 1 - P_{k-1}(\lambda_S) \big)
    ~.
\]

\subsection{Proof of Lemma~\ref{lem:poisson-matching-lp-poly-time}}
\label{app:poisson-matching-lp-poly-time}

The case of $\ell = 0$ holds because the matching LP has at most $|I||J|$ variables and $|I| + |J|$ non-trivial constraints.

Next consider any level $\ell \ge 1$.
The $\ell$-th level Poisson Matching LP has $|I||J|$ variables, and $|I|$ matching constraints for online types.
It remains to give a separation oracle for the exponentially many Poisson constraints.
For any $1 \le m \le \ell$, and any of the $\binom{|J|}{m}$ subsets $T \subseteq J$ with $|T| = m$, we will give a polynomial-time (in $|I|$ and $|J|$) separation oracle for the Poisson constraints:
\begin{equation}
    \label{eqn:poisson-matching-lp-poly-time-subcase}
    \forall S \subseteq I ~: \quad \sum_{i \in S} \sum_{j \in J} x_{ij} \le \sum_{k=1}^m \big( 1 - P_{k-1}(\lambda_S) \big)
    ~.
\end{equation}

Then, combining these $O(|J|^\ell)$ separation oracles gives one for the $\ell$-th level Poisson Matching LP with running time polynomial in $|I|$ and $|J|^\ell$.

\begin{lemma}
    \label{lem:poisson-matching-lp-poly-time-fractional}
    The Poisson constraints~\eqref{eqn:poisson-matching-lp-poly-time-subcase} are equivalent to that for any $0 \le \mu_i \le \lambda_i$, $i \in I$:
    \begin{equation}
        \label{eqn:poisson-matching-lp-poly-time-equivalent}
        \sum_{i \in I} \frac{\mu_i}{\lambda_i} \sum_{j \in J} x_{ij} \le \sum_{k=1}^m \Big( 1 - P_{k-1}\Big(\sum_{i \in I} \mu_i \Big) \Big)
        ~.
    \end{equation}
\end{lemma}

\begin{proof}
    On the one hand, the Poisson constraints~\eqref{eqn:poisson-matching-lp-poly-time-subcase} are special cases of the stated constraint in this lemma, when $\mu_i \in \{0, \lambda_i\}$, $i \in I$.
    On the other hand, the left-hand-side is linear in $\mu_i$, $i \in I$, and the right-hand-side is concave.
    Hence, the difference between the right and the left is minimized at a vertex of the hyperrectangle $\times_{i \in I} [0, \lambda_i]$.
\end{proof}

For any fixed value of $\sum_i \mu_i$, we can maximize the left-hand-side of \eqref{eqn:poisson-matching-lp-poly-time-equivalent} by greedily assigning masses in descending order of $\frac{1}{\lambda_i} \sum_{j \in J} x_{ij}$.
Therefore, it suffices to check constraints~\eqref{eqn:poisson-matching-lp-poly-time-subcase} only for the subsets $S$ of first $k$ elements in $I$ by the descending order of $\frac{1}{\lambda_i} \sum_{j \in J} x_{ij}$ for $1 \le k \le |I|$.

\subsection{Proof of Lemma~\ref{lem:jensen-submodular}}
\label{app:jensen-submodular}

We prove it by a hybrid argument.
Let $\bar{y} = \int_0^1 y(\mu) d\mu$.
For $0 \le k \le d$, define $y^{(k)} : [0, 1] \to [0, 1]^d$ so that for any coordinate $\ell \in [d]$:
\[
    y^{(k)}_\ell(\mu) =
    \begin{cases}
        y_\ell(\mu) & \text{ if $\ell > k$,} \\
        \bar{y}_\ell & \text{ otherwise.}
    \end{cases}
\]

By definition $y^{(0)}(\mu) = y(\mu)$, and $y^{(d)}(\mu) = \bar{y}$ for any $0 \le \mu \le 1$.
It suffices to prove that for any $1 \le k \le d$:
\[
    \int_0^1 f \big(y^{(k-1)}(\mu) \big) d\mu
    \le
    \int_0^1 f \big(y^{k}(\mu) \big) d\mu
    ~.
\]

The lemma then follows by summing this inequality for $1 \le k \le d$.

Next fix any $1 \le k \le d$.
By definition $y^{(k-1)}(\mu), y^{(k)}(\mu)$ differ only in the $k$-th coordinate.
Let:
\[
    \mu^* = \inf \{ \mu : y^{(k-1)}_k(\mu) \ge \bar{y}_k \}
    ~.
\]

Then, for any $0 \le \mu < \mu^*$ the $k$-th coordinate of $y^{(k-1)}(\mu)$ is less than or equal to $\bar{y}_k$, the $k$-th coordinate of $y^{(k)}(\mu)$;
we also have that $y^{(k)}(\mu) \le y^{(k)}(\mu^*)$.
Similarly, for any $\mu^* < \mu < 1$, the $k$-th coordinate of $y^{(k-1)}(\mu)$ is greater than or equal to $\bar{y}_k$, the $k$-th coordinate of $y^{(k)}(\mu)$;
we also have that $y^{(k)}(\mu) \ge y^{(k)}(\mu^*)$.
Combining this with the (anti-)monotonicity of $\partial_k f$, for any $0 \le \mu \le 1$:
\begin{equation}
    \label{eqn:app-jensen-submodular}
    \Big( \bar{y}_k - y^{(k-1)}_k(\mu) \Big) \Big( \partial_k f \big( y^{(k)}(\mu) \big) - \partial_k f \big( y^{(k)}(\mu^*) \big) \Big)
    \ge 0
    ~.
\end{equation}

Hence, by the diminishing returns of $f$ we have:
\begin{align*}
    \int_0^1 f \big(y^{k}(\mu) \big) d\mu - \int_0^1 f \big(y^{(k-1)}(\mu) \big) d\mu
    & 
    = \int_0^1 \Big( f \big(y^{k}(\mu) \big) - f \big(y^{(k-1)}(\mu) \big) \Big) d\mu
    \\%[1ex]
    & 
    \ge \int_0^1 \partial_k f \big( y^{(k)}(\mu) \big) \big( \bar{y}_k - y^{(k-1)}_k(\mu) \big) d\mu
    \tag{$\partial_{kk} f \le 0$} \\%[1ex]
    & 
    \ge \partial_k f \big( y^{(k)}(\mu^*) \big) \int_0^1  \big( \bar{y}_k - y^{(k-1)}_k(\mu) \big) d\mu
    \tag{Eqn.~\eqref{eqn:app-jensen-submodular}}\\[1ex]
    & 
    = 0
    ~.
\end{align*}

\subsection{Proof of Theorem~\ref{thm:2nd-converse-jensen}}
\label{app:2nd-converse-jensen}

For ease of notations we shall write $y_i = \rho_{ij_1}$ and $z_i = \rho_{ij_1} + \rho_{ij_2}$.
They satisfy:
\begin{enumerate}
    \item
    \emph{(Poisson constraint for $y$)}
    For any $S \subseteq I$:
    \[
        \sum_{i \in S} \lambda_i y_i \le 1 - P_0(\lambda_S)
        ~.
    \]
    \item
    \emph{(Poisson constraint for $z$)}
    For any $S \subseteq I$:
    \[
        \sum_{i \in S} \lambda_i z_i \le 2 - P_0(\lambda_S) - P_1(\lambda_S)
        ~.
    \]
    \item
    \emph{(Order constraint)}
    For any $i \in I$:
    \[
        y_i \le z_i
        ~.
    \]
    \item
    \emph{(Fixed sum constraint)}
    \[
        \sum_i \lambda_i y_i = x_{j_1} ~, \quad \sum_i \lambda_i z_i = x_{j_1} + x_{j_2}
        ~.
    \]
\end{enumerate}

This proof does not need the Poisson constraint for $\rho_{j2}$ (i.e., $z_i - y_i$).
We shall next perform a sequence of transformations to the $y_i$'s and $z_i$'s such that $\sum_i f(y_i, z_i)$ is non-increasing after each transformation.

First, sort $I$ in descending order of $y_i$.
We may further assume without loss of generality that $z_i$'s are also in descending order.
Otherwise, suppose that $y_i > y_{i+1}$ but $z_i < z_{i+1}$.
We can change the value of $z_i$ in the first $\min \{ \lambda_i, \lambda_{i+1} \}$ portion of $i$ to $z_{i+1}$, and change the value of $z_{i+1}$ in the last $\min \{ \lambda_i, \lambda_{i+1} \}$ portion of $i+1$ to $z_i$.
By submodularity (Lemma~\ref{lem:submodular}), this weakly decreases the value of $\sum_i f(y_i, z_i)$.

Further, define $\lambda_{\le i} = \sum_{k \le i} \lambda_k$ in the remaining argument for notational simplicity. 
We may assume without loss of generality that there are indices $i_1$ and $i_2$ such that $\lambda_{\le i_1} = \lambda_1^*$ and $\lambda_{\le i_2} = \lambda_2^*$.
This is because we can split any $i$ into two copies with the same $y_i$ and $z_i$ and whose arrival rates sum to $\lambda_i$.

Next we prove the case of $\lambda_1^* \le \lambda_2^*$, which implies $i_1 \le i_2$.
The other case of $\lambda_1^* < \lambda_2^*$ be proved similarly.

We may further assume without loss of generality that both Poisson constraints are tight for any $i \le i_1$, i.e.:
\begin{align*}
    \sum_{k \le i} \lambda_k y_k
    &
    = 1 - P_0 \Big( \lambda_{\le i} \Big)
    ~, \\
    \sum_{k \le i} \lambda_k z_k
    &
    = 2 - P_0 \big( \lambda_{\le i} \big) - P_1 \big( \lambda_{\le i} \big)
    ~.
\end{align*}

Suppose not.
Let $i$ be the smallest index for which the above constraints are not tight.
We may increase the values of $y_i$ and $z_i$ to make them tight, while decreasing the values of $y_{i'}$ and $z_{i'}$ for $i' > i$ to maintain the fixed sum constraint.
By the non-positivity of $f$'s Hessian entries, this weakly decreases the value of $\sum_i \lambda_i y_i$.
Then, for any $i \le i_1$, by the Jensen inequality (Lemma~\ref{lem:jensen-submodular}):
\[
    \lambda_i f(y_i, z_i) \ge \int_{\lambda_{\le i-1}}^{\lambda_{\le i}} f \big( P_0(\lambda), P_1(\lambda) \big) d \lambda
\]

Finally, the tightness of the Poisson constraints for $i \le i_1$ implies $y_i = 0$ for $i > i_1$.
By a similar argument, we may assume without loss of generality that the Poisson constraint for $z$ is tight for any $i > i_1$.
Then, for any $i > i_1$, by the Jensen inequality (Lemma~\ref{lem:jensen-submodular}):
\[
    \lambda_i f(y_i, z_i) \ge \int_{\lambda_{\le i-1}}^{\lambda_{\le i}} f \big( 0, P_1(\lambda) \big) d \lambda
    ~.
\]

Putting together proves the inequality of the theorem.

\subsection{Proof of Lemma~\ref{lem:1st-level-loose-bound}}
\label{app:1st-level-loose-bound}

The lemma holds trivially if $x_j = 0$.
Next consider $x_j > 0$.
Since $p_j(0) = e^{-2tx_j} = 0$, it suffices to prove that for $0 \le t < 1$:
\[
    \frac{d}{dt} \log p_j(t) \ge -2x_j
    ~.
\]

Suppose for contrary that $\{ 0 \le t < 1: \frac{d}{dt} \log p_j(t) < -2x_j \}$ is non-empty.
Let $t_0$ be its infimum.
Then, we have $\frac{d}{dt} \log p_j(t) \ge - 2x_j$ for $0 \le t < t_0$ and therefore $p_j(t) \ge e^{-2tx_j}$ for $0 \le t \le t_0$.
Then, by Eqn.~\eqref{eqn:1st-level-recurrence} and $\log(1-y) \le -y$ for any $y \le 1$, we get that for any $0 \le t \le t_0$:
\[
    \frac{d}{dt} \log p_j(t) \le - x_j
    ~,
\]
and thus:
\[
    p_j(t) \le e^{-tx_j}
    ~.
\]

In particular, the above bound at $t = t_0$ implies that:
\[
    1 - e^{-2x_jt_0} p_j(t_0) \ge 1 - e^{-x_jt_0} \ge 1-\frac{1}{e}
    ~.
\]

Putting into Eqn.~\eqref{eqn:1st-level-recurrence}, we get that:
\[
    \frac{d}{dt} \log p_j(t = t_0) \ge \frac{\log \big(1 - (1-\frac{1}{e})\,x_j \big)}{1-\frac{1}{e}} \ge - \frac{e}{e-1} x_j
    ~.
\]

Since $p_j$ is continuous which further implies the continuity of $\frac{d}{dt} \log p_j(t)$ by our recurrence~\eqref{eqn:1st-level-recurrence}, the above inequality indicates that $\frac{d}{dt} \log p_j(t) \ge - 2 x_j$ holds for $t \in [t_0, t_0+\varepsilon)$ with a sufficiently small $\varepsilon$, contradicting the definition of $t_0$.

\subsection{Proof of Lemma \ref{lem:x-maximum}}
\label{app:x-maximum}

Let $f(x)=\frac{\log(1-x)}{x}$, $g(x)=x\big(1-\frac{e^{-2xt}}{1-yx}\big)$, the left-hand-side is then $(1-yx)f(g(x))$.
Consider:
\[
    h(x)=x\Big(1-e^{-2xt-x\log(1-y)}\Big) \le g(x)
    ~,
\]
which follows by $(1-y)^x \ge 1-yx$.
It holds with equality holds when $x=0$ and $1$.
It remains to prove that $(1-yx)f(h(x))$ achieves its maximum value at $x=1$.

We next show that it is non-decreasing in $x$.
The derivative is:
\begin{equation}
    \label{eqn:x-maximum-derivative}
    \frac{d}{dx} (1-yx)f(h(x)) = (1-yx)f'(h(x))h'(x)-yf(h(x))
    ~,
\end{equation}
where:
\[
    h'(x) = 1+e^{-2xt-x\log(1-y)} \big( 2xt+x\log(1-y)-1 \big)
    ~.
\]

Fix any $x,t$.
When $y$ decreases, $1-yx$ increases, $h'(x)$ increases, $h(x)$ increases, $\frac{f'(h(x))}{f(h(x))}$ increases.
Hence the derivative in Eqn.~\eqref{eqn:x-maximum-derivative} is minimized when $y = 1 - e^{-t}$.
In this case $h(x) = x \big(1-e^{-xt}\big)$ and $h'(x) = 1-(1-xt) e^{-xt}$.
The non-negativity of derivative reduces to:
\[
    \big(1-(1 - e^{-t})x\big) \, \big(1-(1-xt)e^{-xt}\big) \, \frac{f'(x\left(1-e^{-xt}\right))}{f(x\left(1-e^{-xt}\right))}\leq 1 - e^{-t}
    ~.
\]

By $x(1-e^{-xt}) \le 1-e^{-xt}$ and the monotonicity of $\frac{f'(z)}{f(z)}$, it further reduces to:
\[
    \left(\frac{1}{1-e^{-t}}-x\right) \, \big(1-(1-xt)e^{-xt}\big) \, \frac{f'(1-e^{-xt})}{f(1-e^{-xt})} \le 1
    ~.
\]

Next fix the product $z=xt$, let $\alpha(t)=\frac{1}{1-e^{-t}}-x=\frac{1}{1-e^{-t}}-\frac{z}{t}$.
Its derivative is:
\[
    \alpha'(t) = -\frac{e^{-t}}{(1-e^{-t})^2}+\frac{z}{t^2}
    ~.
\]

Therefore:
\[
    e^t(1-e^{-t})^2\alpha'(t) = \frac{2z(\cosh t-1)}{t^2}-1
    ~,
\]
which is increasing.
Hence, $\alpha'(t)$ is at first negative and then non-negative in $t \in [z, 1]$, including always negative and always non-negative as special cases.
This means that the maximum of $\alpha(t)$ is either at $t=z$ or at $t=1$.
It remains to verify:
\begin{align*}
    \left(\frac{1}{1-e^{-z}}-1\right)\cdot(1-(1-z)e^{-z})\cdot\frac{f'(1-e^{-z})}{f(1-e^{-z})}\leq 1.\\
    \left(\frac{1}{1-e^{-1}}-z\right)\cdot(1-(1-z)e^{-z})\cdot\frac{f'(1-e^{-z})}{f(1-e^{-z})}\leq 1.
\end{align*} 

If $z\geq 0.8$, $\frac{1}{1-e^{-z}}-1\leq 0.82$, $\frac{1}{1-e^{-1}}-z\leq 0.78$, while $\frac{f'(1-e^{-z})}{f(1-e^{-z})}\leq \frac{f'(1-e^{-1})}{f(1-e^{-1})}=1.14$.
Hence the inequalities hold.

If $z\leq 0.8$, we have $\frac{1}{1-e^{-z}}-1\geq\frac{1}{1-e^{-1}}-z$.
Hence, it suffices to verify the first inequality.
Expanding $\frac{f'(1-e^{-z})}{f(1-e^{-z})}$, it becomes:
%In other words,
%
\[
    \big(1-(1-z)e^{-z}\big) \, \left(\frac{e^z}{z}-\frac{e^z}{e^z-1}\right) \leq e^z-1
    ~.
\]

Multiplying both sides by $z (e^z-1)$, and viewing it as a quadratic function of $e^z-1$ with coefficients depending on $z$, the inequality follows by:
\[
    \forall z \in [0, 1] : \quad e^z-1 \le \frac{z}{\sqrt{1-z}}
    ~.
\]

\section{Numerical Verification for the Second Level Analysis}
\label{app:2nd-numerical}

This section explains how to numerically lower bound the probability that Poisson OCS matches an offline vertex $j$ by the end when its LP matched level is $x_j = x$.
In other words, we will numerically compute an upper bound of $s_x(1)$.

Let $\Dx, \Dt$ be sufficiently smaller constants such that $\frac{1}{\Dx}, \frac{1}{\Dt}$ are integers.
We will write $[0, 1]_\Delta$ as the set of multiples of $\Delta$ between $0$ and $1$.

\begin{enumerate}
    \item For any $x' \in [0, 1]_\Dx$:
    \begin{enumerate}
        \item Let $x_1 = \max \{x, x'\}$ and $x_2 = \min \{x, x'\}$, and compute $\lambda_1^*, \lambda_2^*$ that satisfy:
        \begin{align*}
            x_1 & = 1 - P_0(\lambda_1^*) ~, \\[1ex]
            x_2 & = \big( 2 - P_0(\lambda_2^*) - P_1(\lambda_2^*) \big) - \big( 1 - P_0(\min\{\lambda_1^*, \lambda_2^*\}) \big) ~.
        \end{align*}
        \item For any time $t \in [0, 1]$ and any $D \in [0, 1]$, define $\Delta \log \hat{d}(D, t)$ as:
        \[
            \Delta \log \hat{d}(D, t)
            =
            - 
            \begin{cases}
                \displaystyle
                \int_0^{\lambda_1^*} \hat{f}_{D, t} \big( P_0(\lambda), P_1(\lambda) \big) d\lambda + \int_{\lambda_1^*}^{\lambda_2^*} \hat{f}_{D, t} \big( 0, P_1(\lambda) \big) d\lambda
                & \mbox{if } \lambda_1^* \le \lambda_2^* ~, \\[3ex]
                \displaystyle
                \int_0^{\lambda_2^*} \hat{f}_{D, t} \big( P_0(\lambda), P_1(\lambda) \big) d\lambda + \int_{\lambda_2^*}^{\lambda_1^*} \hat{f}_{D, t} \big( P_0(\lambda), P_0(\lambda) \big) d\lambda
                & \mbox{if } \lambda_1^* > \lambda_2^* ~,
            \end{cases}
        \]
        where (recall that $z^+ = \max \{z, 0\}$):
        \[
            \hat{f}_{D, t}(y, z) = \frac{e^{t(2x_1+x_2)} D \, y + e^{t(x_1+2x_2)} D \, (z-y)} { \big(e^{t(2x_1+x_2)} D - 1\big)^+ \, y + \big(e^{t(x_1+2x_2)} D - 1 \big)^+\, (z-y) + 1}
            ~.
        \]
        \item For $t \in [0, 1]_\Dt$ recursively compute $\hat{d}_{x_1, x_2}(t)$ by:
        \begin{align*}
            \hat{d}_{x_1, x_2}(0) & = 1 ~, \\[1ex]
            \tilde{d}_{x_1, x_2}(t) & = \hat{d}_{x_1, x_2}(t) \cdot \exp \big( \Dt \cdot \Delta \log \hat{d}(\hat{d}_{x_1, x_2}(t), t) \big) ~, \\[1ex]
            \hat{d}_{x_1, x_2}(t + \Dt) & = \hat{d}_{x_1, x_2}(t) \cdot \exp \big( \Dt \cdot \Delta \log \hat{d}(\tilde{d}_{x_1, x_2}(t), t) \big) ~,
            ~.
        \end{align*}
    \end{enumerate}
    \item For any $t \in [0, 1]_\Dt$, compute:
    \[
        \hat{q}_x(t) = \min \Big\{ e^{-t x}, \max_{x'\in [0, 1]_\Dx} e^{t (x' + \Dx)} \hat{d}_{\max\{x, x'\}, \min\{x, x'\}}(t) \Big\}
        ~.
    \]

    \item For $t \in [0, 1]_\Dt$, recursively compute $\hat{s}_x(t)$ by:
    \begin{align*}
        \hat{s}_x(0) & = 1
        ~; \\[1ex]
        \tilde{s}_x(t) & = \hat{s}_x(t) \exp \big( \Dt \cdot \Delta \log \hat{s}_x(t) \big)
        ~; \\[1ex]
        \hat{s}_x(t + \Delta t) & = \hat{s}_x(t) \exp \big( \Dt \cdot \Delta \log \tilde{s}_x(t) \big)
        ~,
    \end{align*}
    where:
    \begin{align*}
        \Delta \log \hat{s}_x(t)
        &
        = \frac{ \log \left( 1 - x + x \cdot \frac{e^{-t x} \hat{q}_x(t)}{\hat{s}_x(t)} \right)} { 1 - \frac{e^{-t x} \hat{q}_x(t)}{\hat{s}_x(t)}}
        ~,
        \\
        \Delta \log \tilde{s}_x(t)
        &
        = \frac{ \log \left( 1 - x + x \cdot \frac{e^{-t x} \hat{q}_x(t)}{\tilde{s}_x(t)} \right)} { 1 - \frac{e^{-t x} \hat{q}_x(t)}{\tilde{s}_x(t)}}
        ~.
    \end{align*}
\end{enumerate}

\begin{lemma}
    \label{lem:2nd-level-D-hat}
    For any $1 \ge x_1 \ge x_2 \ge 0$ and $t\in [0, 1]_\Dt$, $\hat{d}_{x_1, x_2}(t) \ge d_{x_1, x_2}(t)$.
\end{lemma}
\begin{proof}
    We will prove it by induction in $t \in [0, 1]_\Dt$ by ascending order.
    The base case of $t = 0$ holds with equality because both sides equal $1$.
    Suppose that for some $t < 1$ we have $\hat{d}_{x_1, x_2}(t) \ge d_{x_1, x_2}(t)$.
    Next consider the lemma at time $t+\Dt$.

    First we have that $\tilde{d}_{x_1, x_2}(t) \le \hat{d}_{x_1, x_2}(t)$ by the monotonicity of $\hat{f}_{D,t}(y,z)$ in $D$.
    It implies that $\hat{d}_{x_1, x_2}(t+\Dt) \ge \tilde{d}_{x_1, x_2}(t)$.
    If further $\tilde{d}_{x_1, x_2}(t) \ge d_{x_1, x_2}(t+\Dt)$ we have proved the lemma for $t+\Dt$.

    Otherwise, we can verify that $\Delta \log \hat{d}(\tilde{d}_{x_1, x_2}(t), t) \ge \frac{d}{dt} \log d_{x_1, x_2}(t')$ for any $t \le t' \le t+\Dt$, so by the definition of $\hat{d}_{x_1, x_2}(t+\Dt)$ it is greater than or equal to $d_{x_1, x_2}(t+\Dt)$.
\end{proof}

\begin{lemma}
    \label{lem:2nd-level-Q-hat}
    For any $x\in[0,1]$ and $t\in [0, 1]_\Dt$, $\hat{q}_x(t)\geq q_x(t)$.
\end{lemma}
\begin{proof}
    We first observe that for any $t \in [0, 1]$, $d_{x_1,x_2}(t)$ is non-increasing in both $x_1$ and $x_2$.
    Suppose that $x_1' \ge x_1$ and $x_2' \ge x_2$.
    By the recursion of $d_{x_1, x_2}$ and $d_{x_1', x_2'}$, $d_{x_1, x_2}(t) \le d_{x_1', x_2'}(t)$ would imply:
    \[
        \frac{d}{dt} d_{x_1, x_2}(t) \ge \frac{d}{dt} d_{x_1', x_2'}(t)
        ~.
    \]

    Hence, for any $x' \in [0, 1]_\Dx$ and any $x'' \in [x', x'+\Dx)$, and any time $t \in [0, 1]_\Dt$ we have:
    \begin{align*}
        e^{tx''} d_{\max\{x, x''\}, \min\{x, x''\}}(t)
        &
        \le e^{t(x'+\Dx)} d_{\max\{x, x'\}, \min\{x, x'\}}(t) \\[1ex]
        &
        \le e^{t(x'+\Dx)} \hat{d}_{\max\{x, x'\}, \min\{x, x'\}}(t)
        ~.
        \tag{Lemma~\ref{lem:2nd-level-D-hat}}
    \end{align*}
    
    Combining it the trivial bound of $q_x(t) \le e^{-xt}$ (Lemma~\ref{cor:2nd-level-q-trivial}), the lemma follows from the definition of $\hat{q}_x(t)$.
\end{proof}

Similar to Lemma \ref{lem:2nd-level-Q}, we have the following lemma which underlies the first-level converse Jensen inequality, and the proof is omitted.

\begin{lemma}
    For any $0 \le x \le 1$ and $0 \le t \le 1$, $\hat{s}_x(t) \ge e^{-tx} \hat{q}_x(t)$.
\end{lemma}

\begin{lemma}
    For any $x\in[0,1]$ and $t\in [0, 1]_\Dt$, $\hat{s}_x(t)\geq s_x(t)$.
\end{lemma}

This is the same argument as Lemma~\ref{lem:2nd-level-Q-hat} so we omit the proof.

One can now numerically verify that $\frac{1-\hat{s}_x(1)}{x}\geq 0.716$ for any $x\in[0,1]$ when $\Dx = \Dt = 10^{-4}$, e.g., using \href{https://github.com/h-zhiyi/poisson-ocs/blob/c15622fe8758dfab48762388b4c0e749419db048/numerical/level2_parallel.jl}{our implementation} in Julia.

\end{document}